\documentclass[preprint,11pt,authoryear]{elsarticle}
\usepackage{tikz}
\usepackage{xcolor}
\usepackage{amsfonts}
\usepackage{amsmath}
\usepackage{amsthm}
\usepackage{amssymb}
\usepackage{times}
\usepackage{subfig}
\usepackage{multirow}
\usepackage{setspace}
\usepackage{enumerate}
\usepackage{natbib}
\usepackage{color}
\usepackage{colortbl}
\usepackage{ wasysym }
\usepackage{verbatim}
\usetikzlibrary{shapes,arrows,calc,positioning}
 \usepackage{lscape}
\usepackage{array}
\usepackage{setspace}
\usepackage[normalem]{ulem}
\usepackage{caption}
\usepackage{rotating}

\usepackage{blkarray} 

\makeatletter
\def\ps@pprintTitle{%
    \let\@oddhead\@empty
    \let\@evenhead\@empty
    \def\@oddfoot{\footnotesize\itshape
         {Submitted preprint} \hfill\today}%
    \let\@evenfoot\@oddfoot
    }
\makeatother

\newcolumntype{x}[1]{%
>{\centering\hspace{0pt}}p{#1}}%
\definecolor{Gray}{gray}{0.9}
\newcolumntype{g}{>{\columncolor{Gray}}c}

\def\iid{\buildrel {\rm i.i.d.} \over \sim}

\def\i.i.d.{\buildrel {\rm i.i.d.} \over \sim}

\def\cw#1 { \overset{\mathbb{P}}{\underset{#1}{\longrightarrow}} }
\def\Real{\mathbb{R}}
\def\Natu0{\mathbb{N}_0}

\def \rcov#1#2 {{\rm cov}_{#1}\left( #2\right)}

\DeclareMathOperator*{\argmin}{arg\,min}

\newtheorem{example}{Example}

\newtheorem{lemma}{Lemma}

\newtheorem{proposition}{Proposition}

\newtheorem*{toy*}{Toy Model}

\newtheorem{model}{Model}

\begin{document}
\begin{frontmatter}
\title{Optimal relativities in a modified Bonus-Malus system with\\long memory transition rules and frequency-severity dependence}

\author[EH]{Jae Youn Ahn}
\ead{jaeyahn@ewha.ac.kr}
\author[UN]{Eric C.K. Cheung}
\ead{eric.cheung@unsw.edu.au}
\author[KA]{Rosy Oh}
\ead{rosy.oh5@gmail.com}
\author[UN]{Jae-Kyung Woo$^*$}
\ead{j.k.woo@unsw.edu.au}

\address[EH]{Department of Statistics, Ewha Womans University, Seodaemun-Gu, Seoul 03760, Korea.}
\address[UN]{School of Risk and Actuarial Studies, Australian School of Business, University of New South Wales, Australia.}
\address[KA]{Department of Industrial and Systems Engineering, Korea Advanced Institute of Science and Technology (KAIST), Daejeon 34141, Korea.}
\cortext[cor2]{Corresponding Author}

\begin{abstract}
In the classical Bonus-Malus System (BMS) in automobile insurance, the premium for the next year is adjusted according to the policyholder's claim history (particularly frequency) in the previous year. Some variations of the classical BMS have been considered by taking more of driver's claim experience into account to better assess individual's risk. Nevertheless, we note that in practice it is common for a BMS to adopt transition rules according to the claim history for the past multiple years in countries such as Belgium, Italy, Korea, and Singapore. In this paper, we revisit a modified BMS which was briefly introduced in \citet{L1995} and \citet{PDW03}. Specifically, such a BMS extends the number of Bonus-Malus (BM) levels due to an additional component in the transition rules representing the number of consecutive claim-free years. With the extended BM levels granting more reasonable bonus to careful drivers, this paper investigates the transition rules in a more rigorous manner, and provides the optimal BM relativities under various statistical model assumptions including the frequency random effect model and the dependent collective risk model. Also, numerical analysis of a real data set is provided to compare the classical BMS and our proposed BMS.


\end{abstract}


\end{frontmatter}

\section{Introduction}

From \citet{L1998}, as far as automobile insurance is concerned, insurers in the U.S. and Canada mainly use many of a priori variables (observable factors) to classify the risks of the policyholders while some other countries in Europe, Asia, Latin America and Africa use also other a posteriori rating. The common a posteriori ratemaking process can be done via Bonus-Malus System (BMS), and such a posteriori ratemaking possibly resolves residual heterogeneity still remaining within risk classes. As in Section 1 of \citet{L1998}, although BMS is not commonly used in North America, the situation can change later since regulatory authorities may prohibit insurers from using certain classification variables, and moreover a posteriori rating is known to be an efficient way to assess individual risks.
Based on the policyholder's claim experience, such system discounts premium as a reward of claim-free case (bonus) and charges additional premium in the presence of accidents (malus). This a posteriori ratemaking mechanism takes on various forms across countries but it generally consists of three major components: Bonus-Malus (BM) levels; transition rules; and BM relativities. Each new policyholder is classified into one of the finite number of risk classes depending on his/her observable risk characteristics, and he/she is also assigned an initial BM level. After each policy period, the BM level moves up or down according to the prespecified transition rules penalizing based on the number of reported claims and rewarding those without claim. In a typical BMS, claim experience in the previous policy period is usually utilized to determine the transition rule. However, for better separation of consistently and temporarily good (or bad) risks, we consider to capture claim experience in the multiple previous policy periods in the design of the transition rule like in Belgium. Then, the premium is renewed at a rate which is the product of the BM relativity and the base premium. Here the BM relativity is determined by the BM level which plays a role in the premium adjustment coefficient. In this way, the actual premium charged to policyholders in each BM level could match more closely to their risks reflected in the claim number record for the past few policy periods and consequently the claim costs are fairly shared. In the classical BMS, the base premium is often calculated as a product of the mean frequency and the mean severity due to the assumed independence between these two random variables, e.g. \cite{Pitrebois2003, Denuit2, Chong}. However, the independence assumption is often criticized since some recent insurance reports demonstrate the existence of significant dependence between claim frequency and severity. In the actuarial literature, various statistical models were developed in attempt to capture certain dependency structures between frequency and severity of claims. These include, for example, copula models \citep{Czado, Frees4}, shared or bivariate random effect models \citep{
 Bastida, Czado2015, 
 PengAhn}, and two-part models \citep{Peng, Garrido, Park2018does, AhnValdez2}.

The focus of this paper is two-fold. First, we revisit the extension of the BM levels \citep{L1995, PDW03} and define our long memory transition rules in a more formal setting.
Second, under the proposed transition rules, we consider the optimal relativities not only in a frequency-only model but also in a dependent collective risk model that allows for dependence between frequency and severity.
As an example of the dependent collective model, we consider the copula-based bivariate random effect model as in \citet{PengAhn}, among many choices of shared or bivariate random effect models. 
While the policyholder's BM level goes down if there is no claim in the past year in the classical BMS, the BM level in our model moves to a lower one when he/she does not have any claim history consecutively for years equivalent to the ``period of penalty'' in this modified BMS. Similar case was discussed in Chapter 17 of \cite{L1985} and Section 5 of \cite{PDW03} where the transition is made based on the number of consecutive claim-free years. However, the optimal relativities in this paper are obtained by solving different forms of optimization problems. Under two models, namely the frequency random effect model and the bivariate random effect model, the optimal relativities and the hypothetical mean square error (HMSE) are studied. The rest of the paper is organized as follows. In Section \ref{bm.sec}, the description of two random effect models and the optimal relativity results under the classical BMS rules are given. In Section \ref{mBMS}, a period of the penalty (denoted as ``$pen$'') is introduced in the transition rules such that the BM level goes down only when there is no claim for the last consecutive $(1+pen)$ years. Under the modified transition rules, an ``augmented'' BM level is newly defined so that its transition process has the Markov property. In turn, the optimal relativities are provided under the two different models.
Finally, in Section \ref{num}, a data set from the Wisconsin Local Government Property Insurance Fund (e.g. \cite{Frees4}) is utilized to illustrate the impact of the long memory transition rules on our BMS compared to the classical BMS.

\section{Model descriptions and basic results}\label{bm.sec}



For the $i$-th policyholder in the $t$-th policy year, we let $N_{it}$ be the number of claims and $\boldsymbol{Y}_{it}=(Y_{it1}, \ldots, Y_{itN_{it}})$ be the vector of associated claim amounts, where $Y_{itj}$ is the $j$-th claim amount. (Note that $\boldsymbol{Y}_{it}$ is undefined when $N_{it}=0$.) Also, the aggregate claim amount and average claim amount are defined as
\begin{equation}\label{sumaverageDef}
S_{it}:=\begin{cases}
\sum_{j=1}^{N_{it}}Y_{itj}, & N_{it}>0\\
0, &N_{it}=0
\end{cases}
\quad\quad
\hbox{and}
\quad\quad
M_{it}:=\begin{cases}
\frac{1}{N_{it}}\sum_{j=1}^{N_{it}}Y_{itj}, & N_{it}>0,\\
\hbox{0}, &N_{it}=0,
\end{cases}
\end{equation}
respectively. Clearly, these two quantities are directly linked as $S_{it}=N_{it}M_{it}$. 

In the insurance ratemarking process, a priori ratemaking results in risk classification of the (new) policyholder whose past claim history is insufficient  for the ratemaking purposes. Using the observed risk characteristics of the policyholders (collected in the row vector $\boldsymbol{X}_i$ for the $i$-th policyholder), their risk classes are determined and the a priori premiums are also obtained. Then, the residual heterogeneity is explained by capturing the unobserved risk characteristics of policyholders (denoted as $\Theta_i$ for the $i$-th policyholder) in a posteriori ratemaking process and consequently their premium amounts are adjusted based on the claim history. Here, $\boldsymbol{X}_i$ and $\Theta_i$ do not depend on time. In general, the frequency part may be considered enough to construct the BMS under the assumption of independence between frequency and severity. However, recent BMS research in \citet{PengAhn} reveals that the dependence between frequency and severity plays a critical role in the determination of optimal relativities in case of significant dependence. In this regards, the superscripts [1] and [2] representing frequency and severity components respectively are introduced to variables and their corresponding realizations. When there is only the frequency component, we drop the notation [1] for convenience. To model the frequency and the severity of claims, the techniques of Generalized Linear Model (GLM) in \cite{deJong2008} with the assumption of exponential dispersion family for the random components will be applied (see \ref{EDF} for the descriptions).

\subsection{Random effect models}\label{REM}

For the unobserved risk characteristics in the a posteriori ratemating, let us describe two different models which take into account the randomness of frequency only or both frequency and severity. For convenience, we assume that there are $\mathcal{K}$ risk classes predetermined at the moment of when the contracts begin. We start with the descriptions of the model for the frequency part only.
\begin{model}[Frequency model with random effect]\label{mod.1}
First, the weight of the $k$-th risk class is defined as
  \begin{equation}\label{eq.401}
  w_k:= \mathbb{P}(\boldsymbol{X}=\boldsymbol{x}_{k}), \qquad k = 1,\ldots, \mathcal{K},
  \end{equation}
where $\boldsymbol{X}$ is the row vector of the observed risk characteristics of a randomly picked person from the population of policyholders, and $\boldsymbol{x}_{k}$ is a vector of the observed risk characteristics for the $k$-th risk class. Note that $w_k$ can be regarded as the fraction of policyholders with observed risk characteristics $\boldsymbol{x}_{k}$.
The a priori premium for the $i$-th policyholder is determined as $  \Lambda_{i}:=\eta^{-1} (\boldsymbol{X}_{i}\boldsymbol{\beta})$, where $\eta(\cdot)$ is a link function, $\boldsymbol{X}_i$ is the row vector of observed risk characteristics of this $i$-th policyholder, and $\boldsymbol{\beta}$ is a column vector of regression coefficients in the count regression model estimated from the past data concerning the number of reported claims. Incorporating unobserved risk component for the frequency\footnote{In practice, it is very hard to figure out all the possible combinations of the observed risk characteristics for the model. Hence, introduction of unobserved risk characteristics $\Theta_i$ is essential, and this provides theoretical and practical justifications for the a posteriori risk classification.}, the conditional distribution of the number of claims $N_{it}$ for the $i$-th policyholder in the $t$-th policy year given the observed risk characteristics $\boldsymbol{X}_{i}=\boldsymbol{x}_{i}$ and unobserved risk characteristics  $\Theta_i=\theta_i$ along with the distribution of $\Theta_i$ is provided in \ref{model1}.
\end{model}

The following model is an extension of two-part model \citep{Frees, Garrido} by introducing a bivariate random effect to model various types of dependence such as:\vspace{-.05in}
\begin{itemize}
  \item dependence among frequencies;\vspace{-.05in}
  \item dependence among severities; and \vspace{-.05in}
  \item dependence between frequency and severity.
\end{itemize}\vspace{-.1in}

\begin{model}[Collective risk model with bivariate random effect \citep{PengAhn}]\label{mod.2}


Recall that the superscripts [1] and [2] represent the frequency and the severity components respectively. As an extension of Model 1 above, the random component of the claim severity is also taken into account for a priori and a posteriori ratemaking processes. In contrast to (\ref{eq.401}), the weight of the $k$-th risk class is defined as
  \begin{equation*}
  w_k:= \mathbb{P}({\boldsymbol{X}^{[1]}=\boldsymbol{x}_{k}^{[1]},
  \boldsymbol{X}^{[2]}=\boldsymbol{x}_{k}^{[2]}
}), \qquad k = 1,\ldots, \mathcal{K},
  \end{equation*}
  where $(\boldsymbol{X}^{[1]},\boldsymbol{X}^{[2]})$ is the pair of row vectors containing the observed risk characteristics relevant to the frequency and severity of claims for a randomly picked person from the population of policyholders, and $(\boldsymbol{x}_k^{[1]},\boldsymbol{x}_k^{[2]})$ is the corresponding pair for the $k$-th risk class. The a priori premium for the $i$-th policyholder can be obtained from the information of the vector $(\Lambda_i^{[1]},\Lambda_i^{[2]})$, which is in turn determined by the observed risks characteristics as
\[
        \Lambda_{i}^{[1]}:=\eta_1^{-1} (\boldsymbol{X}^{[1]}_{i}\boldsymbol{\beta}^{[1]})
        \quad\hbox{and}\quad \Lambda_{i}^{[2]}:=\eta_2^{-1}(\boldsymbol{X}^{[2]}_{i}\boldsymbol{\beta}^{[2]}),
  \]
where $\eta_1(\cdot)$ and $\eta_2(\cdot)$ are link functions, $\boldsymbol{X}^{[1]}_i$ and $\boldsymbol{X}^{[2]}_i$ are the $i$-th policyholder's row vectors of observed risk characteristics for frequency and severity, and $\boldsymbol{\beta}^{[1]}$ and $\boldsymbol{\beta}^{[2]}$ are column vectors of parameters to be estimated. More importantly, the dependence between claim frequency and severity is also considered as the corresponding unobserved risk characteristics $(\Theta_i^{[1]},\Theta_i^{[2]})$ are modeled via a bivariate copula. More specific distributional descriptions of this model are provided in \ref{model2}.
\end{model}

It is instructive to note that when Model 2 is restricted to the claim frequency only, it is conceptually equivalent to Model 1.

\subsection{Optimal relativities}\label{OR}

In the classical BMS framework, common transition rules are such that the BM level is lowered by one for a claim-free year and increased by $h$ levels per claim, leading to what is referred to as the $-1/+h$ system. Based on such transition rules, the BM level for each policyholder evolves as a (discrete-time) Markov Chain. Let us denote $L$ as a random variable representing the BM level (from 0 to $z$) for a randomly picked policyholder from the population in a stationary state (independent of time). Its distribution (that is, the proportion of policyholders in each level) is determined by the frequency component as
\begin{equation}\label{statProb}
\mathbb{P}(L=\ell)
 =\sum_{k=1}^{\mathcal{K}} w_{k} \int \pi_\ell(\lambda_k\theta, \psi) g(\theta){\rm d}\theta, \qquad \ell=0, \ldots, {z},
\end{equation}
where $\pi_\ell(\lambda_k\theta, \psi)$ is the stationary probability that a policyholder with expected frequency $\lambda_k\theta$ is in level $\ell$. (It is understood that $\lambda_k=\eta^{-1} (\boldsymbol{x}_{k}\boldsymbol{\beta})$, $g(\cdot)$ and $\psi$ in \eqref{statProb} are replaced by $\lambda^{[1]}_k=\eta_1^{-1} (\boldsymbol{x}^{[1]}_{k}\boldsymbol{\beta}^{[1]})$, $g_1(\cdot)$ and $\psi^{[1]}$ respectively if one considers Model 2.) We shall provide explanations on how to obtain the stationary probabilities at the beginning of Section \ref{SEC32} for our more general model in relation to Example \ref{ex.2}. The relativity associated with the BM level $\ell$ is denoted by ${\zeta}(\ell)$. The following lemmas review how to determine the optimal relativity $\tilde{\zeta}(\ell)$ in two optimization settings under the model assumptions in Section \ref{REM}.

\begin{lemma}\normalfont \citep{Chong} \label{lem.1}
Consider the optimization problem
\begin{equation}\label{eq.ahn13}
(\tilde{\zeta}(0), \ldots, \tilde{\zeta}(z)):=\argmin_{({\zeta}(0), \ldots, {\zeta}(z))\in \Real^{{z+1}}} \mathbb{E}[(\Lambda\Theta-\Lambda{\zeta}(L))^2]
\end{equation}
under the frequency-only Model \ref{mod.1}. Because $\mathbb{E}[N_{it}|\Lambda_i,\Theta_i]=\Lambda_i\Theta_i$ (see \eqref{eq.402}) is the ``correct'' premium for the $i$-th policyholder if we knew $\Theta_i$, one can regard $\Lambda\Theta=\eta^{-1} (\boldsymbol{X}\boldsymbol{\beta})\Theta$ as the ``correct'' premium for a policyholder randomly picked from the population having observed risk characteristics $\boldsymbol{X}$ and unobserved risk characteristics $\Theta$. The optimization \eqref{eq.ahn13} amounts to choosing the relativities to minimize the mean squared difference between $\Lambda\Theta$ and the actual premium charged $\Lambda{\zeta}(L)$ when a policyholder is in BM level $L$. Then, the optimal relativities can be analytically calculated as
\begin{equation*}
\tilde{\zeta}(\ell):=\frac{\mathbb{E}[\Lambda^2 \Theta |L=\ell]}{\mathbb{E}[\Lambda^2| L=\ell]}
=\frac{ \sum_{k=1}^{\mathcal{K}} w_k \lambda_k^2\int  \theta  \pi_\ell(\lambda_k \theta,\psi) g(\theta){\rm d}\theta}
{\sum_{k=1}^{\mathcal{K}} w_k \lambda_k^2\int \pi_\ell(\lambda_k\theta,\psi)g(\theta){\rm d}\theta}
, \qquad \ell=0, \ldots, {z}.
\end{equation*}
\end{lemma}

\begin{lemma}\normalfont \citep{PengAhn} \label{lem.2}
Consider the optimization problem
\begin{equation}\label{opt2}
(\tilde{\zeta}(0), \ldots, \tilde{\zeta}(z)):=\argmin_{({\zeta}(0), \ldots, {\zeta}(z))\in \Real^{{z+1}}} \mathbb{E}[(\Lambda^{[1]}\Lambda^{[2]}\Theta^{[1]}\Theta^{[2]}-\Lambda^{[1]}\Lambda^{[2]}{\zeta}(L))^2]
\end{equation}
under the frequency-severity Model \ref{mod.2}. As $\mathbb{E}[S_{it}|\Lambda_i^{[1]},\Lambda_i^{[2]},\Theta_i^{[1]},\Theta_i^{[2]}]=\Lambda_i^{[1]}\Lambda_i^{[2]}\Theta_i^{[1]}\Theta_i^{[2]}$ (see \ref{model2}) is the ``correct'' premium for the $i$-th policyholder, $\Lambda^{[1]}\Lambda^{[2]}\Theta^{[1]}\Theta^{[2]}=\eta_1^{-1} (\boldsymbol{X}^{[1]}\boldsymbol{\beta}^{[1]}) \eta_2^{-1}(\boldsymbol{X}^{[2]}\boldsymbol{\beta}^{[2]})\Theta^{[1]}\Theta^{[2]}$ is the ``correct'' premium for a policyholder randomly picked from the population having observed risk characteristics $(\boldsymbol{X}^{[1]},\boldsymbol{X}^{[2]})$ and unobserved risk characteristics $(\Theta^{[1]},\Theta^{[2]})$. The optimization \eqref{opt2} is concerned with choosing the relativities to minimize the mean squared difference between $\Lambda^{[1]}\Lambda^{[2]}\Theta^{[1]}\Theta^{[2]}$ and the actual premium charged $\Lambda^{[1]}\Lambda^{[2]}{\zeta}(L)$ when a policyholder is in BM level $L$. Then, the optimal relativities can be analytically calculated as
\begin{align*}
\tilde{\zeta}(\ell):=&~\frac{\mathbb{E}[( \Lambda^{[1]}\Lambda^{[2]})^2  \Theta^{[1]}\Theta^{[2]}|L=\ell]}{\mathbb{E}[(\Lambda^{[1]} \Lambda^{[2]})^2|L=\ell]}\\
=&~\frac{\sum_{k=1}^{\mathcal{K}} w_{k} (\lambda_{k}^{[1]}\lambda_{k}^{[2]})^2  \int\int \theta^{[1]}  \theta^{[2]} \pi_\ell(\lambda_{k}^{[1]}\theta^{[1]},\psi^{[1]}) h(\theta^{[1]},\theta^{[2]}){\rm d}\theta^{[1]}{\rm d}\theta^{[2]}}{\sum_{k=1}^{\mathcal{K}} w_{k} (\lambda_{k}^{[1]}\lambda_{k}^{[2]})^2  \int \pi_\ell(\lambda_{k}^{[1]}\theta^{[1]},\psi^{[1]})g_1(\theta^{[1]}){\rm d}\theta^{[1]}}, \qquad \ell=0, \ldots, {z}.
\end{align*}
\end{lemma}

The optimal relativities in Lemma \ref{lem.1} and \ref{lem.2} can be applied to
various transition rules and BM levels $\ell$. For the calculation of optimal relativity in the BMS literature,  the Markov property of the transition rules is typically utilized, and the BM level is often assumed to be stationary \citep{Pitrebois2003, Denuit2, Chong, PengAhn}. The most common (and also one of the simplest) BMS is the $-1/+h$ system, while other variations are also available, for example, in \cite{Chong}. However, many countries including Belgium, Korea, and Singapore adapt transition rules with long memory. While such transition rules do not directly possess Markov property, there are examples in the literature showing how these can be dealt with \citep{L1995, Pitrebois2003}. Key technique involves augmentation of BM levels, and the following section generalizes these examples.

We remark that the optimization in \eqref{eq.ahn13} or \eqref{opt2} is not the only optimization criterion that can be used. There are various criteria to define the optimality of BM relativities, e.g. \cite{Denuit2, Chong}. More importantly, BM relativities obtained from \eqref{eq.ahn13}, \eqref{opt2}, and \cite{Denuit2, Chong} are known to create some systematic bias in the prediction of premium. Such systematic bias is called the double-counting problem, and we refer interested readers to \cite{L1995, Taylor1997, Ohdouble} for the details of the double-counting problem and its solution.

\section{A modified BMS with augmented BM levels}\label{mBMS}

In this section, the number of consecutive claim-free years for a policyholder is taken into account in the transition rules. Assuming there are $z+1$ BM levels, we let $L_{t}\in\{0, 1, \ldots, z\}$ be the BM level at time $t\in\{0,1,\ldots\}$ so that $L_{t}$ is the BM level applicable for the $(t+1)$-th policy year, that is, from time $t$ to time $t+1$. Similar to previous notation, for $t\in\{1,2,\ldots\}$ we use $N_{t}$ to denote the number of claims in the $t$-th policy year.

\subsection{Modified transition rules of $-1/+h/pen$}

While the classical BMS usually adopts the so-called $-1/+h$ system for $h\le z$ as described in Section \ref{OR}, we shall introduce an additional component, namely ``$pen$'', to this system. Specifically, under such a $-1/+h/pen$ system, each reported claim increases the BM level by $h$ while $1+pen^*_t$ consecutive claim-free years will result in decrease of one BM level at time $t+1$ when moving from the $(t+1)$-th policy year to the $(t+2)$-th. Here $pen_t^*$ is defined as
\begin{equation}\label{pstar}
pen_t^*:=\min\{pen, t\},
\end{equation}
for $t\in\{0,1,\ldots\}$. This implies that the BM level at time $t=0$ starts with no penalty, and indeed a new policyholder can have his/her BM level reduced by one every year as long as he/she maintains a no-claim record. (This is also equivalent to letting $N_0=N_{-1}=\ldots=N_{-pen}=0$ and saying that one requires $1+pen$ consecutive claim-free years to decrease BM level by one.) Note that the classical BMS is retrieved from our extended model by letting $pen=0$. We further assume that a new policyholder without any driving history belongs to the BM level $\ell_0$ in the beginning. Then, $L_{t}$ is mathematically represented as
\begin{equation}\label{eq.8}
L_t:=
\begin{cases}
  \min\{L_{t-1}+h N_{t}, z\}, & N_{t}>0,\\
  \max\{L_{t-1}-1, 0\}, &
  N_{t}= \ldots =N_{\max\{t-pen^*_{t-1}, 1\}} =0,\\
  L_{t-1}, &\hbox{otherwise},
\end{cases}
\end{equation}
for $t\in\{1,2,\ldots\}$, with $L_0=\ell_0$. It is also assumed that the BM relativity ${\zeta}(\ell)$ is applied in the same manner as in the classical $-1/+h$ system. That is, each BM level $\ell$ is bestowed with BM relativity ${\zeta}(\ell)$. Compared to the classical $-1/+h$ system, an additional $pen^*_t$ claim-free years are required to reduce BM level by one, whereas the penalty of climbing $h$ BM levels per claim remains the same. Such a system is introduced to motivate policyholders to drive carefully by applying stricter rule of lowering the BM level (reward case) while keeping the same rule of increasing the BM level (penalty case). This kind of the transition rules are adopted in some countries such as Singapore and Korea.

While the Markov property in the transition probability appearing in certain BMS is convenient for the analysis of BMS-related problems (e.g. \cite{Taylor1997, Pitrebois2003, Denuit2}), it is obvious that with the current definition of $L_t$, the transition rules are not Markovian as $L_t$ depends on not only $L_{t-1}$ but also the numbers of claims in the past multiple periods. To resolve this issue, as explained in \cite{L1995} and \cite{Pitrebois2003}, fictitious levels can be included to redefine the state space of the BM levels. Subdividing some of the BM levels in order to include the information of the number of claim-free years results in an increase of the number of states compared to the classical case. More precisely, for $\ell=h, \ldots, z$, the BM level $\ell$ is augmented into
\begin{equation*}
(\ell)_0, (\ell)_1, \ldots, (\ell)_{pen},
\end{equation*}
where the subscript stands for the number of additional claim-free periods (compared to the classical BMS) required to get rewarded, while $\ell\in \{0, \ldots, h-1\}$ is just translated into $(\ell)_0$ without augmentation.\footnote{No augmentation of the BM levels $0, \ldots, h-1$ is necessary because these BM levels can only be reached as a result of reward from the previous years.}
Hence, the original $(z+1)$ BM levels in the $-1/+h/pen$ system have been augmented to $[h+(z-h+1)\times (pen+1)]$ BM levels. Let us denote the BM level at time $t$ under the augmented system as $L_t^*$. Then, it is convenient to define the possible combinations of $(\ell)_a$ in our proposed model (that is, the state space) as
\begin{equation}\label{newstatespace}
\mathcal{A}_{z, pen}:=\{ (\ell)_0\big\vert \ell=0, \ldots, h-1\}\cup \{ (\ell)_0, \ldots, (\ell)_{pen} \big\vert \ell=h, \ldots, z \}.
\end{equation}
Furthermore, suppose that a policyholder who was at the BM level $L_{t-1}^*=(\ell)_a\in \mathcal{A}_{z, pen}$ at time $t-1$ has reported at time $t$ that $N_{t}$ accidents occurred during the $t$-th policy year. Then the new BM level $L_{t}^*$ at time $t$ is determined as
\begin{equation}\label{eq.6}
L_t^*:=
\begin{cases}
  ( \max\{\ell-1, 0 \})_0, & \hbox{if}\quad  N_{t}=0 \,\,\hbox{and}\,\,a=0,\\
  ( \ell)_{a-1}, & \hbox{if}\quad  N_{t}=0 \,\,\hbox{and}\,\,a\neq 0,\\
  (\min\{\ell+h  N_{t}, z \})_{pen}, & \hbox{if}\quad  N_{t}>0,\\
\end{cases}
\end{equation}
for $t\in\{1,2,\ldots\}$, where a new policyholder without driving history starts with $L_0^*=(\ell_0)_0$. The transitions of $L_t^*$ are now Markovian, as $L_t^*$ depends on $L_{t-1}^*=(\ell)_a$ and the latest claim number only. Also, we assume that the BM relativity in state $(\ell)_a\in \mathcal{A}_{z, pen}$ under the augmented system defined via (\ref{eq.6}) is the same as the one in $\ell\in\{0, 1, \ldots, z\}$ for the original system defined in (\ref{eq.8}) in the sense that the relativities depend on the BM level $\ell$ but not the information $a$ which is artificially introduced to make the transitions Markovian. In other words, under the proposed model one has the relativities
\begin{equation}\label{eq.21}
\zeta^*\left( (\ell)_a \right):=\zeta(\ell),\qquad (\ell)_a\in \mathcal{A}_{z, pen}.
\end{equation}
We call the BMS with BM levels \eqref{newstatespace}, transition rules \eqref{eq.6} and relativities \eqref{eq.21} an augmented $-1/+h/pen$ system. The following example illustrates the transition rules with a certain penalty period.

\begin{example}\label{ex.1}\normalfont
Let us consider the $-1/+2/2$ system with $z=20$. That is, the BM level is to be increased by two per claim while having $1+pen^*_t$ consecutive claim-free years is to be rewarded by a decrease of one BM level. We assume that a new policyholder without driving history starts at BM level $L_0=10$ (so $\ell=10$ at time 0). Given the BM level $L_{t}$ at time $t$ and the claim frequency $N_{t+1}$ in the $(t+1)$-th year, the next year's BM level $L_{t+1}$ according to (\ref{eq.8}) and alternatively $L^\ast_{t+1}$ according to (\ref{eq.6}) are summarized in Table \ref{tb.1}. (Note that, with $pen=2$, \eqref{pstar} implies that $pen_t^*$ stays level at 2 once $t$ reaches 2.) The first few transitions are explained as follows:\vspace{-.05in}
  \begin{enumerate}
    \item[i.] Transition from $t=0$ to $t=1$:\\
     At time 0 one has $pen^*_0=0$ from (\ref{pstar}), and therefore $1+pen^*_0=1$ claim-free year is required to reduce the BM level by one at time 1 (like the classical BMS) and so we set $a=0$ at time 0. As it turns out that there is no claim reported at $t=1$ (i.e. $N_1=0$), the BM level is reduced from $L_0=10$ to $L_1=9$ (and thus $\ell=9$ at $t=1$). Now $pen^*_1=1$ according to (\ref{pstar}), meaning that a total of $1+pen^*_1=2$ consecutive claim-free years are needed to reduce one BM level next year at time 2. Since there has already been one claim-free year, only one more claim-free year is needed to earn such a reward (again like the classical BMS) and therefore one sets $a=0$ at time 1.\vspace{-.05in}
    \item[ii.] Transition from $t=1$ to $t=2$:\\
   Since $N_2=0$ again, the condition of reward is satisfied and the BM level moves down to $L_2=8$ (and $\ell=8$) at time $2$. With $pen^*_2=2$ at time 2, it is known that $1+pen^*_t=3$ consecutive claim-free years are needed to lower BM level by one. With two consecutive claim-free years already in hand, we need only one more claim-free year to reduce BM level in the next year and therefore $a=0$ at time 2.\vspace{-.05in}
   \item[iii.] Transition from $t=2$ to $t=3$:\\
     Since $N_3=2$ in the third year, the BM level moves up to 12 ($=8+2\times 2$) at time 3 (and so $\ell=12$). As $pen^*_3=2$ at time 3, one needs $1+pen^*_3=3$ consecutive claim-free years to reduce BM level by one, which means two claim-free years in additional to the one claim-free year required in the classical BMS. Consequently, we have $a=2$ at $t=3$.
     \vspace{-.05in}
     \item[iv.] Transition from $t=3$ to $t=4$:\\
     Because $N_4=1$ in the forth year, the BM level increases to $L_4=14$ ($=12+2\times 1$) at time 4.  As there is a claim in this period, three consecutive claim-free years are still needed to reduce BM level by one. So, we have $a=2$ at $t=4$.
          \vspace{-.05in}
     \item[v.] Transition from $t=4$ to $t=5$:\\
     Although there is no claim in the fifth year as $N_5=0$, the condition of having three consecutive claim-free years in order to enjoy a reward is not satisfied, and therefore the BM level stays at the same level as last year so that $L_5=14$. Then two more consecutive claim-free years are required to lower the BM level by one. This is one additional year compared to the classical $-1/+2$ system so $a=1$ at $t=5$.
  \end{enumerate}

  \begin{table}[h]
  \begin{center}
  \begin{tabular}{|c|c|c|c|c|c|c|c|c|c|c|}
    \hline
    $t$ & 0 & 1 & 2 & 3 & 4 & 5 & 6 & 7 & 8 & 9 \\
    \hline\hline
    $N_{t+1}$ & 0 & 0 & 2 & 1 & 0 & 0 & 0 & 0 & 1 &  \\
    $L_t$ & 10 & 9 & 8 & 12 & 14 & 14 & 14 & 13 & 12 & 14 \\
    $pen_t^*$ & 0 & 1 & 2 & 2 & 2 & 2 & 2& 2 & 2 & 2 \\
    $L_t^*$ & $(10)_0$ & $(9)_0$ & $(8)_0$ & $(12)_2$ & $(14)_2$ & $(14)_1$ & $(14)_0$ & $(13)_0$ & $(12)_0$ & $(14)_2$ \\
    \hline
  \end{tabular}
  \caption{Transition rules under a $-1/+2/2$ system}\label{tb.1}
  \end{center}
  \end{table}
\end{example}\vspace{-.3in}



\begin{example}\normalfont\label{ex.2}

Suppose that we have the $-1/+2/1$ system with $z=7$. We consider Model \ref{mod.1} where a policyholder, say the $i$-th policyholder, has observed and unobserved risk characteristics $\boldsymbol{X}_i=\boldsymbol{x}_i$ and ${\Theta}_i={\theta}_i$ respectively. Then, based on the augmented form of BM levels defined in \eqref{newstatespace}, one finds a total of $(2+(7-2+1)\times 2)=14$ BM levels with the transition matrix $\boldsymbol{P}$ given by

{\begin{footnotesize}
\[
\begin{blockarray}{*{14}{c} l}
&(0)_0 & (1)_0 & (2)_0 & (2)_1 & (3)_0 & (3)_1& (4)_0 & (4)_1& (5)_0 & (5)_1& (6)_0 & (6)_1& (7)_0 & (7)_1\\
\begin{block}{c(*{13}{c} l)}
(0)_0&  p_0& 0  & 0  & p_1& 0  & 0  & 0  & p_2& 0  & 0  & 0  & p_3& 0  & 1-p_0-p_1-p_2-p_3 \\
(1)_0&  p_0& 0  & 0  & 0  & 0  & p_1& 0  & 0  & 0  & p_2& 0  & 0  & 0  & 1-p_0-p_1-p_2\\
(2)_0&  0  & p_0& 0  & 0  & 0  & 0  & 0  & p_1& 0  & 0  & 0  & p_2& 0  & 1-p_0-p_1-p_2\\
(2)_1&  0  & 0  & p_0& 0  & 0  & 0  & 0  & p_1& 0  & 0  & 0  & p_2& 0  & 1-p_0-p_1-p_2\\
(3)_0&  0  & 0  & p_0& 0  & 0  & 0  & 0  & 0  & 0  & p_1& 0  & 0  & 0  & 1-p_0-p_1\\
(3)_1&  0  & 0  & 0  & 0  & p_0& 0  & 0  & 0  & 0  & p_1& 0  & 0  & 0  & 1-p_0-p_1\\
(4)_0&  0  & 0  & 0  & 0  & p_0& 0  & 0  & 0  & 0  & 0  & 0  & p_1& 0  & 1-p_0-p_1\\
(4)_1&  0  & 0  & 0  & 0  & 0  & 0  & p_0& 0  & 0  & 0  & 0  & p_1& 0  & 1-p_0-p_1\\
(5)_0&  0  & 0  & 0  & 0  & 0  & 0  & p_0& 0  & 0  & 0  & 0  & 0  & 0  & 1-p_0\\
(5)_1&  0  & 0  & 0  & 0  & 0  & 0  & 0  & 0  & p_0& 0  & 0  & 0  & 0  & 1-p_0\\
(6)_0&  0  & 0  & 0  & 0  & 0  & 0  & 0  & 0  & p_0& 0  & 0  & 0  & 0  & 1-p_0\\
(6)_1&  0  & 0  & 0  & 0  & 0  & 0  & 0  & 0  & 0  & 0  & p_0& 0  & 0  & 1-p_0\\
(7)_0&  0  & 0  & 0  & 0  & 0  & 0  & 0  & 0  & 0  & 0  & p_0& 0  & 0  & 1-p_0\\
(7)_1&  0  & 0  & 0  & 0  & 0  & 0  & 0  & 0  & 0  & 0  & 0  & 0  & p_0& 1-p_0\\
\end{block}
\end{blockarray}
 \]
 \end{footnotesize} }
where $p_n$ in the matrix is defined by (see \eqref{eq.402})
\[
p_n:=\mathbb{P}(N_{it}=n|\Theta_i=\theta_i, \boldsymbol{X}_i=\boldsymbol{x}_i), \qquad n=0,1,\ldots.
\]
For simplicity, the dependence of $p_n$ on $\theta_i$ and $\boldsymbol{x}_i$ is suppressed. Similarly, if we assume Model \ref{mod.2} instead and the $i$-th policyholder has observed and unobserved risk characteristics given by $(\boldsymbol{X}^{[1]}_i, \boldsymbol{X}^{[2]}_i)=(\boldsymbol{x}^{[1]}_i, \boldsymbol{x}^{[2]}_i)$ and $(\Theta^{[1]}_i, \Theta^{[2]}_i)=(\theta^{[1]}_i, \theta^{[2]}_i)$, then $p_n$ in the matrix $\boldsymbol{P}$ is (see \eqref{eq.n})
\[
p_n:=\mathbb{P}(N_{it}=n|\Theta^{[1]}_i=\theta^{[1]}_i, \boldsymbol{X}^{[1]}_i=\boldsymbol{x}^{[1]}_i), \qquad n=0,1,\ldots.
\]
%
%
\end{example}

\subsection{Optimal relativities}\label{SEC32} 


In this section, under a $-1/+h/pen$ system, the optimal values of the relativities $\zeta(\ell)$ for $\ell=0,1,\ldots,z$ are studied. Following the classical BMS literature, we denote the stationary distribution of $L_t^*$ by $L^*$ for a randomly picked policyholder. Then, under the representation in Model \ref{mod.1}, one finds in an analogous manner to \eqref{statProb} that
\begin{equation}\label{eq.ahn5}
\mathbb{P}(L^*=(\ell)_a)
=\sum_{k=1}^\mathcal{K} w_{k} \int \pi_{(\ell)_a}^*(\lambda_k\theta, \psi) g(\theta){\rm d}\theta,\qquad (\ell)_a\in \mathcal{A}_{z, pen},
\end{equation}
where $\pi_{(\ell)_a}^*(\lambda_k\theta, \psi)$ is the stationary distribution for a policyholder to be in BM level $(\ell)_a$ given that his/her a priori claim frequency is $\lambda_k$ and unobserved risk characteristics are summarized in $\theta$. (Again, $\lambda_k$, $g(\cdot)$ and $\psi$ are replaced by $\lambda^{[1]}_k$, $g_1(\cdot)$ and $\psi^{[1]}$ respectively if we consider Model 2 instead.) For example, the 14-dimensional row vector of stationary probabilities corresponding to Example \ref{ex.2}, namely
\[
\boldsymbol{\pi}^*(\lambda_{k}\theta,\psi):=(\pi_{(0)_0}^*(\lambda_{k}\theta,\psi), \pi_{(1)_0}^*(\lambda_{k}\theta,\psi),\ldots, \pi_{(7)_1}^*(\lambda_{k}\theta,\psi)),
\]
can be obtained as the solution of
\begin{equation*}
\begin{cases}
  \boldsymbol{\pi}^*(\lambda_{k}\theta,\psi) = \boldsymbol{\pi}^*(\lambda_{k}\theta,\psi)\boldsymbol{P},\\
  \boldsymbol{\pi}^*(\lambda_{k}\theta,\psi)\boldsymbol{e}_{14}=1,\\
\end{cases}
\end{equation*}
where $\boldsymbol{P}$ is the transition matrix in Example \ref{ex.2} (calculated with $\boldsymbol{x}_{k}$ and $\theta$ in place of $\boldsymbol{x}_{i}$ and $\theta_i$ respectively), and $\boldsymbol{e}_{14}$ is a 14-dimensional column vector of ones.


Under the augmented system, we consider
the optimal relativities as the solution of the optimization problem
 \begin{equation}\label{eq.406}
(\tilde{\zeta}(0), \ldots, \tilde{\zeta}(z)):=\argmin_{({\zeta}(0), \ldots, {\zeta}(z))\in \Real^{{z+1}}}\mathbb{E}[(\Lambda\Theta-\Lambda {\zeta}^*(L^*))^2]
 \end{equation}
resembling \eqref{eq.ahn13} where we aim to predict the frequency under the frequency-only Model \ref{mod.1}. 
On the other hand, when we are interested in the prediction of the aggregate claim, like  \eqref{opt2} the optimal relativities are given by the solution of the optimization problem
 \begin{equation}\label{eq.4066}
(\tilde{\zeta}(0), \ldots, \tilde{\zeta}(z)):=\argmin_{({\zeta}(0), \ldots, {\zeta}(z))\in \Real^{{z+1}}} \mathbb{E}[(\Lambda^{[1]}\Lambda^{[2]}\Theta^{[1]}\Theta^{[2]}-\Lambda^{[1]}\Lambda^{[2]} {\zeta}^*(L^*))^2]
 \end{equation}
under the frequency-severity Model \ref{mod.2}. 
Recall that, by definition, $\zeta^*:\mathcal{A}_{z, pen}\mapsto \Real $ in \eqref{eq.406} or \eqref{eq.4066} is completely characterized by $\zeta(\ell)$ for $\ell=0,\ldots,z$ (see \eqref{eq.21}).



\begin{proposition}\label{thm.oh.1}\normalfont
Consider the $-1/+h/pen$ system under Model \ref{mod.1}. The optimal relativities 
defined as the solution of the optimization problem \eqref{eq.406} are given by
\begin{equation}\label{eq.409}
\tilde{\zeta}(\ell):=\frac{\mathbb{E}[\Lambda^2  \Theta | L^*=(\ell)_0]}
{\mathbb{E}[\Lambda^2 | L^*=(\ell)_0]}, \qquad \ell= 0, \ldots,{h-1},
\end{equation}
and
  \begin{equation}\label{eq.407}
\tilde{\zeta}(\ell):=\frac{\sum_{a=0}^{pen}\mathbb{E}[\Lambda^2  \Theta|  L^*=(\ell)_a]\mathbb{P}(L^*=(\ell)_a)}
{\sum_{a=0}^{pen}\mathbb{E}[\Lambda^2 | L^*=(\ell)_a]\mathbb{P}(L^*=(\ell)_a)},
\qquad \ell=h, \ldots,{z}; a=0, \ldots, pen,
\end{equation}
where the numerator and denominator can be calculated based on
\begin{equation}\label{numeratorProp1part1}
\mathbb{E}[\Lambda^2 \Theta | L^*=(\ell)_a]
=\frac{\sum_{k=1}^\mathcal{K} w_{k} \lambda_{k}^2  \int \theta  \pi_{(\ell)_a}^*(\lambda_{k}\theta, \psi) g(\theta){\rm d}\theta}
{\mathbb{P}(L^*=(\ell)_a)}
\end{equation}
and
\begin{equation}\label{eq.4077}
\mathbb{E}[\Lambda^2 | L^*=(\ell)_a] = \frac{\sum_{k=1}^\mathcal{K} w_{k} \lambda_{k}^2  \int \pi_{(\ell)_a}^*(\lambda_{k}\theta, \psi)g(\theta){\rm d}\theta}
{\mathbb{P}(L^*=(\ell)_a)},
\end{equation}
respectively. Note that when $pen=0$ the optimal relativity $\tilde{\zeta}(\ell)$ reduces to that in \citet{Chong} given in Lemma \ref{lem.1}.
\end{proposition}

\begin{proof}
Note that the objective function in \eqref{eq.406} can be written as
\begin{align}\label{prop1proofstep1}
\mathbb{E}[(\Lambda\Theta-\Lambda {\zeta}^*(L^*))^2]=&\sum_{(\ell)_a\in \mathcal{A}_{z, pen}}\mathbb{E}[(\Lambda\Theta-\Lambda {\zeta}^*((\ell)_a))^2|L^*=(\ell)_a]\mathbb{P}(L^*=(\ell)_a)\nonumber\\
=&\sum_{(\ell)_a\in \mathcal{A}_{z, pen}}\mathbb{E}[(\Lambda\Theta-\Lambda {\zeta}(\ell))^2|L^*=(\ell)_a]\mathbb{P}(L^*=(\ell)_a),
\end{align}
where the last line follows from \eqref{eq.21}. We observe from \eqref{newstatespace} that:\vspace{-.05in}
\begin{itemize}
  \item for each $\ell=0, \ldots, h-1$, the relativity $\zeta(\ell)$ is shared by only one BM level $(\ell)_0$; and\vspace{-.05in}
  \item for each $\ell=h, \ldots,z$, the same relativity $\zeta(\ell)$ is shared by the augmented BM levels $(\ell)_0, \ldots, (\ell)_{pen}$.
\end{itemize}\vspace{-.05in}
Consequently, \eqref{prop1proofstep1} becomes
\begin{align}\label{prop1proofstep2}
\mathbb{E}[(\Lambda\Theta-\Lambda {\zeta}^*(L^*))^2] =&~\sum_{\ell=0}^{h-1}\mathbb{E}[(\Lambda\Theta-\Lambda {\zeta}(\ell))^2|L^*=(\ell)_0]\mathbb{P}(L^*=(\ell)_0)\nonumber\\
&+\sum_{\ell=h}^z\sum_{a=0}^{pen}\mathbb{E}[(\Lambda\Theta-\Lambda {\zeta}(\ell))^2|L^*=(\ell)_a]\mathbb{P}(L^*=(\ell)_a).
\end{align}
For each fixed $\ell=0, \ldots, h-1$, differentiation of \eqref{prop1proofstep2} with respect to ${\zeta}(\ell)$ for optimization yields
\begin{equation*}
\mathbb{E}[-2\Lambda(\Lambda\Theta-\Lambda\zeta(\ell))|L^*=(\ell)_0]\mathbb{P}(L^*=(\ell)_0)=0,
\end{equation*}
from which \eqref{eq.409} follows. On the other hand, when $\ell=h, \ldots,{z}$, taking derivative with respect to ${\zeta}(\ell)$ in \eqref{prop1proofstep2} leads to
\begin{equation*}
\sum_{a=0}^{pen}\mathbb{E}[-2\Lambda(\Lambda\Theta-\Lambda\zeta(\ell))|L^*=(\ell)_a]\mathbb{P}(L^*=(\ell)_a)=0,
\end{equation*}
proving \eqref{eq.407}.



Here, the numerators of \eqref{eq.409} and \eqref{eq.407} can be calculated using
\begin{align}\label{numeratorProp1}
\mathbb{E}[\Lambda^2\Theta| L^*=(\ell)_a] =&~\frac{1}{\mathbb{P}(L^*=(\ell)_a)} \mathbb{E}[\Lambda^2\Theta I(L^*=(\ell)_a)] \nonumber\\
=&~\frac{1}{\mathbb{P}(L^*=(\ell)_a)} \sum_{k=1}^\mathcal{K}\int \lambda_k^2\theta\mathbb{P}(\Theta\in {\rm d}\theta, \Lambda=\lambda_k,L^*=(\ell)_a)\nonumber\\
=&~\frac{1}{\mathbb{P}(L^*=(\ell)_a)} \sum_{k=1}^\mathcal{K}\int \lambda_k^2\theta\mathbb{P}(L^*=(\ell)_a|\Theta\in {\rm d}\theta, \Lambda=\lambda_k) \mathbb{P}(\Theta\in {\rm d}\theta, \Lambda=\lambda_k).
\end{align}
Since the a priori  $\boldsymbol{X}$ and the a posteriori $\Theta$ are independent and $\Lambda=\eta^{-1} (\boldsymbol{X}\boldsymbol{\beta})$, one has that $\mathbb{P}(\Theta\in {\rm d}\theta, \Lambda=\lambda_k)=\mathbb{P}(\Lambda=\lambda_k)\mathbb{P}(\Theta\in {\rm d}\theta)=w_kg(\theta){\rm d}\theta$. Using this together with the fact that $\mathbb{P}(L^*=(\ell)_a|\Theta\in {\rm d}\theta, \Lambda=\lambda_k)=\pi_{(\ell)_a}^*(\lambda_{k}\theta, \psi)$ confirms that \eqref{numeratorProp1} reduces to \eqref{numeratorProp1part1}. The expectation \eqref{eq.4077} can be obtained in almost identical manner and the details are omitted.
\end{proof}

\begin{proposition}\label{prop.2}\normalfont
Consider the $-1/+h/pen$ system under Model \ref{mod.2}. The optimal relativities 
defined as the solution of the optimization problem \eqref{eq.4066} are given by
  \begin{equation*}
\tilde{\zeta}(\ell):=\frac{\mathbb{E}[(\Lambda^{[1]}\Lambda^{[2]})^2  \Theta^{[1]}\Theta^{[2]}|L^*=(\ell)_0]}
{\mathbb{E}[(\Lambda^{[1]} \Lambda^{[2]})^2| L^*=(\ell)_0]},
\qquad \ell=0, \ldots,{h-1},
\end{equation*}
and
  \begin{equation*}
\tilde{\zeta}(\ell):=\frac{\sum_{a=0}^{pen}\mathbb{E}[( \Lambda^{[1]}\Lambda^{[2]})^2  \Theta^{[1]}\Theta^{[2]}| L^*=(\ell)_a]\mathbb{P}(L^*=(\ell)_a)}
{\sum_{a=0}^{pen}\mathbb{E}[(\Lambda^{[1]} \Lambda^{[2]})^2| L^*=(\ell)_a]\mathbb{P}(L^*=(\ell)_a)},
\qquad \ell=h, \ldots,{z}; a=0, \ldots, pen,
\end{equation*}
where the numerator and denominator can be calculated based on
\[
\mathbb{E}[( \Lambda^{[1]}\Lambda^{[2]})^2 \Theta^{[1]}\Theta^{[2]}| L^*=(\ell)_a]=\frac{\sum_{k=1}^\mathcal{K} w_{k} (\lambda_{k}^{[1]}\lambda_{k}^{[2]})^2  \int\int \theta^{[1]}\theta^{[2]}  \pi_{(\ell)_a}^*(\lambda_{k}^{[1]}\theta^{[1]},\psi^{[1]}) h(\theta^{[1]}, \theta^{[2]}){\rm d}\theta^{[1]}{\rm d} \theta^{[2]}} {\mathbb{P}(L^*=(\ell)_a)}
\]
and
\[
\mathbb{E}[( \Lambda^{[1]}\Lambda^{[2]})^2| L^*=(\ell)_a] = \frac{\sum_{k=1}^\mathcal{K} w_{k} (\lambda_{k}^{[1]}\lambda_{k}^{[2]})^2  \int \pi_{(\ell)_a}^*(\lambda_{k}^{[1]}\theta^{[1]},\psi^{[1]})g_1(\theta^{[1]}){\rm d}\theta^{[1]}} {\mathbb{P}(L^*=(\ell)_a)},
\]
respectively. Note that when $pen=0$ the optimal relativity $\tilde{\zeta}(\ell)$ reduces to that in \citet{PengAhn} given in Lemma \ref{lem.2}.
\end{proposition}
\begin{proof}
The proof of Proposition \ref{prop.2} is similar to that of Proposition \ref{thm.oh.1}, and thus we omit the details.
\end{proof}

\subsection{Numerical Analysis}

Here a simple example is provided to examine the effect of the period of penalty on the optimal BM relativity $\tilde{\zeta}(\ell)$ and stationary probability $\mathbb{P}(L=\ell)$. First, it is convenient to introduce the relation
\begin{equation}\label{PrLl}
\mathbb{P}(L=\ell)=\sum_{\{a|(\ell)_a\in \mathcal{A}_{z, pen}\}} \mathbb{P}(L^*=(\ell)_a),
\end{equation}
which follows from the definition of $L^*$. The performances of our modified BM transition rules under a $-1/+1/pen$ system as well as a $-1/+2/pen$ system for $pen=0,1,2,3$ are compared by computing the hypothetical mean square error (HMSE). For Model 1 (the frequency random effect model) with the set of BM relativities $\boldsymbol{\zeta}=\{\zeta(\ell)\}_{\ell=0}^z$, 
the HMSE is expressed as
 \begin{align}\label{HMSEDefModel1}
 {\rm HMSE}(\boldsymbol{\zeta})
&:=\mathbb{E}[(\Lambda\Theta-\Lambda \zeta^*(L^*))^2]\nonumber\\
&=\sum_{k=1}^\mathcal{K}\sum_{(\ell)_a\in \mathcal{A}_{z, pen}}  w_k
\int  (\lambda_k \theta-\lambda_k \zeta(\ell))^2
\pi_{(\ell)_a}^*(\lambda_k\theta,\psi) g(\theta) {\rm d}\theta.
\end{align}
If $\boldsymbol{\zeta}$ is replaced by the vector of optimal relativities $\tilde{\boldsymbol{\zeta}}=\{\tilde{\zeta}(\ell)\}_{\ell=0}^z$ calculated using Proposition \ref{thm.oh.1}, then ${\rm HMSE}(\tilde{\boldsymbol{\zeta}})$ is the minimized value corresponding to the right-hand side of the optimization problem \eqref{eq.406}.

Similarly, for Model \ref{mod.2} with the set of BM relativities $\boldsymbol{\zeta}=\{\zeta(\ell)\}_{\ell=0}^z$, following the logic in \citet{PengAhn}, the HMSE is expressed as
 \begin{align*}
 {\rm HMSE}(\boldsymbol{\zeta})
&:=\mathbb{E}[(\Lambda^{[1]}\Lambda^{[2]}\Theta^{[1]}\Theta^{[2]}-\Lambda^{[1]}\Lambda^{[2]} \zeta^*(L^*))^2]\\
&=\sum_{k=1}^\mathcal{K}\sum_{(\ell)_a\in \mathcal{A}_{z, pen}}  w_k
\int \int  \big(\lambda_k^{[1]}\lambda_k^{[2]} \theta^{[1]}\theta^{[2]}-
\lambda_k^{[1]}\lambda_k^{[2]} \zeta(\ell) \big)^2
\pi_{(\ell)_a}^*(\lambda_k\theta,\psi^{[1]}) h(\theta^{[1]}, \theta^{[2]}) {\rm d}\theta^{[1]} {\rm d}\theta^{[2]}.
\end{align*}
Again, ${\rm HMSE}(\tilde{\boldsymbol{\zeta}})$ is the minimum on the right-hand side of the optimization problem \eqref{eq.4066} when $\tilde{\boldsymbol{\zeta}}=\{\tilde{\zeta}(\ell)\}_{\ell=0}^z$ calculated using Proposition \ref{prop.2}.

\begin{example}\normalfont \label{ex.3}
Under the frequency model with random effect described as Model \ref{mod.1}, it is assumed that there is only one risk class, where the frequency $N$ follows a Poisson distribution with a priori frequency $\lambda=\lambda_0$ (so that the distribution function $F$ can be put in the form of \eqref{eq.402}) and the random effect $\Theta$ in \eqref{eq.403} has a Lognormal distribution with mean 1 and $\sigma^2=0.99$, that is,
\begin{align*}
N|(\Lambda=\lambda,\Theta=\theta) &\sim {\rm Poisson} (\lambda\theta) \qquad {\rm with} \quad \lambda=\lambda_0,\\
\Theta &\sim {\rm Lognormal}(-\sigma^2/2, \sigma^2) \qquad {\rm with} \quad \sigma^2=0.99.
\end{align*}
With $z=9$ and $pen=0,1,2,3$, the stationary probabilities $\mathbb{P}(L=\ell)$ and the optimal BM relativities $\tilde{\zeta}(\ell)$ for all $\ell$'s and the resulting HMSE (under the optimal relativities) are calculated under a $-1/+1/pen$ system and a $-1/+2/pen$ system. The results when $\lambda_0=0.05$ are first summarized in Tables \ref{table.1}(a) and (b).

In obtaining the tables, we first construct the transition probability matrix according to the transition rules of the BMS and the claim frequency distribution with the assumed parameters (see Example \ref{ex.2}) and obtain the stationary probability $\pi_{(\ell)_a}^*$ as explained at the beginning of Section \ref{SEC32}. Then, $\mathbb{P}(L=\ell)$ can be calculated from (\ref{eq.ahn5}) and (\ref{PrLl}). Note that integration with respect to $\theta$ under the distributional assumption on the density $g(\theta)$ of $\Theta$ is needed in (\ref{eq.ahn5}). While explicit evaluation of such one dimensional integral is not possible in general, numerical integration can be easily implemented in most computing software. In turn, the optimal relativities $\tilde{\boldsymbol{\zeta}}=\{\tilde{\zeta}(\ell)\}_{\ell=0}^z$ can be obtained from Proposition \ref{thm.oh.1}, where numerical integration is again performed based on \eqref{numeratorProp1part1} and \eqref{eq.4077} (but the integral appearing in \eqref{eq.4077} is the same as that in (\ref{eq.ahn5})). Calculation of HMSE can be similarly done via \eqref{HMSEDefModel1} once one has calculated the optimal relativities $\tilde{\boldsymbol{\zeta}}$.

We first look at the $-1/+1/pen$ system in Table \ref{table.1}(a). For each fixed value of $pen=0, 1, 2, 3$, most of the population is in the lowest BM level 0 (as $\mathbb{P}(L=0)$ is at least $82\%$), which can be explained by the relatively low mean claim frequency of $\lambda_0=0.05$. By inspecting $\mathbb{P}(L=\ell)$ for $pen=0, 1, 2, 3$, we observe that an increase in $pen$ tends to move some of the population toward higher BM levels, leading to a diversification of the stationary BM levels. This is because a higher $pen$ means that more claim-free years are needed to reduce the BM level by one, rendering it more difficult for drivers at a high BM level to move downward. 
As mentioned in Proposition \ref{thm.oh.1}, the optimal relativities under $pen=0$ are the same as those in \citet{Chong}. As $pen$ increases, it is important to observe that the optimal relativity $\tilde{\zeta}(\ell)$ for each fixed BM level $\ell$ decreases. Although a lower BM relativity means that a driver occupying a given BM level pays less expensive premium in a $-1/+1/pen$ system with a higher $pen$, the insurer's premium income is in turn compensated by an increased portion of drivers occupying higher BM levels as $pen$ increases. A higher $pen$ also results in an improvement of the HMSE in this example.
Similar patterns can be found in $-1/+2/pen$ system as shown in Table \ref{table.1}(b). 
However, for fixed $pen = 1, 2, 3$, the optimal relativity $\tilde{\zeta}(\ell)$ in Table \ref{table.1}(b) is not necessarily increasing in $\ell$ especially for low BM levels. For example, it is observed that $\tilde{\zeta}(2)$ is slightly smaller than $\tilde{\zeta}(1)$. With $h=2$, a possible explanation for $\tilde{\zeta}(2)<\tilde{\zeta}(1)$ is that certain policyholders occupying BM level 2 could be those who were in the lowest BM level 0 (and had one claim) in the previous year, and these are still good drivers. On the other hand, those in BM level 1 could not have come directly from BM level 0. It is interesting to note that, although $\tilde{\zeta}(\ell)$ is not always increasing in $\ell$ in Table \ref{table.1}(b), reporting $n\ge 1$ claims moves a policyholder upward by $2n$ BM levels and this always leads to an increase in $\tilde{\zeta}(\ell)$. Nevertheless, potential problems arise when a policyholder moves down from BM level 2 to BM level 1 after a number of claim-free years and then suffers an increase in the optimal relativity (and hence premium). For future work, constrained optimization may be considered by requiring $\tilde{\zeta}(\ell)$ to be increasing in $\ell$, but it is unlikely that this will lead to optimal relativities that admit expressions as explicit as those in Propositions \ref{thm.oh.1} and \ref{prop.2}.

When $\lambda_0=1$, we consider BM levels up to $z=14$ because the assumption of a higher mean claim frequency will result in a larger proportion of policyholders occupying higher BM levels. The $-1/+1/pen$ system in Table \ref{table.1}(c) shows some similar features: as $pen$ increases, $\tilde{\zeta}(\ell)$ decreases and the distribution (at equilibrium) of policyholders shifts from lower BM levels to higher ones. However, since a significant fraction of the population is already in the highest BM level 14 when $pen=0$, increasing $pen$ anti-diversifies the stationary distribution of BM levels. We observe that HMSE is increased with higher $pen$.

We also notice one interesting phenomenon when comparing the cases where the claim frequencies are of different magnitude. In Tables \ref{table.1}(a) and (b) for which $\lambda_0=0.05$, the optimal relativity $\tilde{\zeta}(0)$ for BM level 0 is still at least 69\% while that for the highest BM level can be as high as 17. An intuitive explanation is that, given the low claim frequency $\lambda_0=0.05$, it is not unusual that drivers do not file claims. As a result, claim-free drivers (those occupying the lowest BM level) are not necessary much better drivers than others and therefore they are not rewarded with much premium discount. On the other hand, drivers occupying the highest BM level must have reported some claims in the past years, and they are more likely to be worse drivers with more risky unobserved risk characteristics (that is, higher $\Theta$) and thus more severely penalized. Moving to Table \ref{table.1}(c) for which $\lambda_0=1$, it becomes much more common for drivers to report claims and occupy higher BM levels, so they are not penalized too much with the highest BM relativity $\tilde{\zeta}(14)$ around 2. But it is those at lower BM levels who have claim-free years that should be regarded as significantly better drivers, and they are given great discount with the BM relativity $\tilde{\zeta}(0)$ no larger than 24.6\%.

From this example, we conclude that while the extension of the classical $-1/+h$ system to the $-1/+h/pen$ system provides room for improvement of the prediction power via lowering the HMSE, such an improvement is not always guaranteed.
\end{example}

\begin{example}\label{ex.4} \normalfont
Under the collective risk model with random effect described as Model \ref{mod.2}, we assume one risk class only, where the distribution function $F_1$ in \eqref{eq.n} for the frequency part follows a Poisson distribution with $\lambda^{[1]}=\lambda_0^{[1]}$ and the distribution function $F_2$ in \eqref{eq.s} for the severity part follows a Gamma distribution with $\lambda^{[2]}=\lambda_0^{[2]}$ and $\psi^{[2]}=\psi_0^{[2]}$.
The marginal distributions of the random effects $(\Theta^{[1]}, \Theta^{[2]})$ are specified as
\[
\begin{cases}
\Theta^{[1]} ~\sim &{\rm Lognormal}(-\sigma_1^2/2, \sigma_1^2) \qquad {\rm with} \quad \sigma_1^2=0.99,\\
\Theta^{[2]} ~\sim &{\rm Lognormal}(-\sigma_2^2/2, \sigma_2^2) \qquad {\rm with} \quad \sigma_2^2=0.29,
\end{cases}
\]
with distribution functions $G_1$ and $G_2$ respectively, and the dependence is described by a Gaussian copula $C$ with correlation coefficient $\rho=-0.45$ (so that the joint distribution function is given by $H=C(G_1, G_2)$ in \eqref{eq.4}).
With $pen=0,1,2,3$, we have calculated the values of the stationary probabilities $\mathbb{P}(L=\ell)$ and the optimal BM relativities $\tilde{\zeta}(\ell)$ for all $\ell$'s together with the HMSE under a $-1/+1/pen$ system for the cases $z=9$ and $\lambda_0^{[1]}=0.05$ (Table \ref{table.2}(a)) as well as $z=14$ and $\lambda_0^{[1]}=1$ (Table \ref{table.2}(c)) and under a $-1/+2/pen$ system for the case $z=9$ and $\lambda_0^{[1]}=0.05$ (Table \ref{table.2}(b)). In constructing Table \ref{table.2}, the parameters $\lambda_0^{[2]}=\exp(8.00)$, and $\psi_0^{[2]}=1/0.670$ are kept fixed. Note that the values of HMSE in Table \ref{table.2} are much higher than those in Table \ref{table.1} concerning Example \ref{ex.3}, and this is simply because of the fact that the HMSE in the present example measures the error concerning the aggregate claim (and the individual severity has large magnitude as $\lambda_0^{[2]}=\exp(8.00)$) while the HMSE in Example \ref{ex.3} is about the error in the claim number. Moreover, some of the optimal relativities in Table \ref{table.2} are quite different from those in Table \ref{table.1}, showing the importance of taking into account the dependence between claim frequency and claim severity for modeling purposes. Nonetheless, similar patterns to those in Table \ref{table.1} are observed in Table \ref{table.2}. Note also that the stationary probabilities in Table \ref{table.2} are identical to their counterparts in Table \ref{table.1} because the probabilities only depend on the distribution of the claim frequency.
\end{example}
\vspace{-.2in}


\begin{table}[h!]
  \caption{(Example 3) Distribution of $L$ and optimal relativities under the frequency-only Model \ref{mod.1}} 
  \label{table.1}\vspace{-.05in}
  \centering
  \resizebox{0.9\textwidth}{!}{ 
 \begin{tabular}{ l c c c c c c c c c c c c c c c c c c c }
 \multicolumn{12}{l}{(a) $-1/+1/pen$ systems with various $pen$ for $\lambda_0=0.05$}\\
 \hline
&$pen$ &  \multicolumn{2}{c}{0} && \multicolumn{2}{c}{1} && \multicolumn{2}{c}{2} && \multicolumn{2}{c}{3} \\\cline{3-4} \cline{6-7}\cline{9-10}\cline{12-13}
Level $\ell$ && $\tilde{\zeta}(\ell)$ & $\mathbb{P}(L=\ell)$ && $\tilde{\zeta}(\ell)$ & $\mathbb{P}(L=\ell)$ && $\tilde{\zeta}(\ell)$ & $\mathbb{P}(L=\ell)$ &&  $\tilde{\zeta}(\ell)$ & $\mathbb{P}(L=\ell)$  \\
 \hline
9	&&	17.096	&	0.001	&&	12.635	&	0.002	&&	10.288	&	0.004	&&	8.801	&	0.008	\\
8	&&	13.947	&	0.000	&&	9.626	&	0.001	&&	7.468	&	0.002	&&	6.151	&	0.002	\\
7	&&	11.978	&	0.000	&&	8.413	&	0.001	&&	6.597	&	0.002	&&	5.477	&	0.002	\\
6	&&	10.384	&	0.000	&&	7.331	&	0.001	&&	5.786	&	0.002	&&	4.832	&	0.003	\\
5	&&	9.408	&	0.001	&&	6.282	&	0.001	&&	4.985	&	0.003	&&	4.187	&	0.004	\\
4	&&	7.304	&	0.001	&&	5.197	&	0.002	&&	4.159	&	0.004	&&	3.520	&	0.006	\\
3	&&	5.552	&	0.002	&&	4.035	&	0.005	&&	3.284	&	0.008	&&	2.815	&	0.012	\\
2	&&	3.676	&	0.007	&&	2.816	&	0.015	&&	2.367	&	0.023	&&	2.076	&	0.030	\\
1	&&	2.003	&	0.044	&&	1.676	&	0.073	&&	1.483	&	0.095	&&	1.348	&	0.111	\\
0	&&	0.887	&	0.944	&&	0.826	&	0.898	&&	0.779	&	0.858	&&	0.740	&	0.821	\\
\hline\hline																			
HMSE	&& \multicolumn{2}{c}{0.00287}	&&	\multicolumn{2}{c}{0.00249}&&	\multicolumn{2}{c}{0.00227}&&	\multicolumn{2}{c}{0.00213}	\\ \hline
\end{tabular}
} 
  \bigskip

\resizebox{0.9\textwidth}{!}{ 
 \begin{tabular}{ l c c c c c c c c c c c c c c c c c c c }
 \multicolumn{12}{l}{(b) $-1/+2/pen$ systems with various $pen$ for $\lambda_0=0.05$}\\
 \hline
&$pen$ &  \multicolumn{2}{c}{0} && \multicolumn{2}{c}{1} && \multicolumn{2}{c}{2} && \multicolumn{2}{c}{3}  \\\cline{3-4} \cline{6-7}\cline{9-10}\cline{12-13}
 Level $\ell$ && $\tilde{\zeta}(\ell)$ & $\mathbb{P}(L=\ell)$ && $\tilde{\zeta}(\ell)$ & $\mathbb{P}(L=\ell)$ && $\tilde{\zeta}(\ell)$ & $\mathbb{P}(L=\ell)$ &&  $\tilde{\zeta}(\ell)$ & $\mathbb{P}(L=\ell)$ \\
 \hline
9	&&	10.288	&	0.002	&&	8.105	&	0.007	&&	6.858	&	0.013	&&	6.030	&	0.020	\\
8	&&	8.106	&	0.002	&&	5.793	&	0.004	&&	4.567	&	0.005	&&	3.807	&	0.007	\\
7	&&	6.755	&	0.002	&&	5.101	&	0.004	&&	4.209	&	0.005	&&	3.633	&	0.006	\\
6	&&	5.410	&	0.003	&&	4.006	&	0.005	&&	3.263	&	0.008	&&	2.796	&	0.011	\\
5	&&	4.638	&	0.003	&&	3.674	&	0.005	&&	3.135	&	0.007	&&	2.772	&	0.008	\\
4	&&	3.349	&	0.007	&&	2.599	&	0.014	&&	2.205	&	0.020	&&	1.950	&	0.027	\\
3	&&	2.973	&	0.009	&&	2.497	&	0.012	&&	2.206	&	0.013	&&	2.000	&	0.014	\\
2	&&	1.779	&	0.041	&&	1.508	&	0.066	&&	1.351	&	0.086	&&	1.241	&	0.101	\\
1	&&	1.677	&	0.038	&&	1.512	&	0.034	&&	1.394	&	0.030	&&	1.303	&	0.028	\\
0	&&	0.811	&	0.892	&&	0.762	&	0.850	&&	0.723	&	0.812	&&	0.691	&	0.778	\\
\hline\hline	
HMSE	&& \multicolumn{2}{c}{0.00260}	&&	\multicolumn{2}{c}{0.00244}&&	\multicolumn{2}{c}{0.00235}&&	\multicolumn{2}{c}{0.00229}	\\ \hline
\end{tabular}
}

  \bigskip

\resizebox{0.9\textwidth}{!}{ 
 \begin{tabular}{ l c c c c c c c c c c c c c c c c c c c }
 \multicolumn{12}{l}{(c) $-1/+1/pen$ systems with various $pen$ for $\lambda_0=1$}\\

 \hline
&$pen$ &  \multicolumn{2}{c}{0} && \multicolumn{2}{c}{1} && \multicolumn{2}{c}{2} && \multicolumn{2}{c}{3}  \\\cline{3-4} \cline{6-7}\cline{9-10}\cline{12-13}
 Level $\ell$ && $\tilde{\zeta}(\ell)$ & $\mathbb{P}(L=\ell)$ && $\tilde{\zeta}(\ell)$ & $\mathbb{P}(L=\ell)$ && $\tilde{\zeta}(\ell)$ & $\mathbb{P}(L=\ell)$ &&  $\tilde{\zeta}(\ell)$ & $\mathbb{P}(L=\ell)$ \\
 \hline
14	&&	2.092	&	0.311	&&	1.713	&	0.463	&&	1.528	&	0.562	&&	1.416	&	0.634	\\
13	&&	1.147	&	0.089	&&	0.791	&	0.060	&&	0.618	&	0.045	&&	0.512	&	0.035	\\
12	&&	0.914	&	0.044	&&	0.677	&	0.036	&&	0.546	&	0.030	&&	0.462	&	0.025	\\
11	&&	0.794	&	0.027	&&	0.602	&	0.025	&&	0.494	&	0.022	&&	0.422	&	0.020	\\
10	&&	0.718	&	0.020	&&	0.548	&	0.019	&&	0.453	&	0.018	&&	0.390	&	0.016	\\
9	&&	0.662	&	0.016	&&	0.505	&	0.016	&&	0.419	&	0.015	&&	0.362	&	0.014	\\
8	&&	0.618	&	0.014	&&	0.470	&	0.015	&&	0.390	&	0.014	&&	0.337	&	0.013	\\
7	&&	0.579	&	0.014	&&	0.439	&	0.014	&&	0.363	&	0.014	&&	0.314	&	0.013	\\
6	&&	0.543	&	0.014	&&	0.410	&	0.015	&&	0.339	&	0.014	&&	0.293	&	0.013	\\
5	&&	0.508	&	0.016	&&	0.382	&	0.016	&&	0.315	&	0.016	&&	0.273	&	0.014	\\
4	&&	0.470	&	0.020	&&	0.353	&	0.020	&&	0.291	&	0.018	&&	0.252	&	0.017	\\
3	&&	0.429	&	0.027	&&	0.322	&	0.026	&&	0.266	&	0.024	&&	0.230	&	0.021	\\
2	&&	0.380	&	0.042	&&	0.286	&	0.038	&&	0.238	&	0.033	&&	0.207	&	0.029	\\
1	&&	0.321	&	0.078	&&	0.246	&	0.065	&&	0.206	&	0.053	&&	0.180	&	0.044	\\
0	&&	0.246	&	0.269	&&	0.197	&	0.172	&&	0.169	&	0.122	&&	0.150	&	0.091	\\
\hline\hline																			
HMSE	&& \multicolumn{2}{c}{1.08320}	&&	\multicolumn{2}{c}{1.23048}&&	\multicolumn{2}{c}{1.32241}&&	\multicolumn{2}{c}{1.38605}	\\ \hline
\end{tabular}
} 
\end{table}


\begin{table}[h!]
  \caption{(Example 4) Distribution of $L$ and optimal relativities under the frequency-severity Model \ref{mod.2} with dependence} \label{table.2}\vspace{-.05in}
  \centering
  \resizebox{0.9\textwidth}{!}{ 
 \begin{tabular}{ l c c c c c c c c c c c c c c c c c c c }
 \multicolumn{12}{l}{(a) $-1/+1/pen$ systems with various $pen$ for $\lambda_0^{[1]}=0.05$}\\
 \hline
&$pen$ &  \multicolumn{2}{c}{0} && \multicolumn{2}{c}{1} && \multicolumn{2}{c}{2} && \multicolumn{2}{c}{3} \\\cline{3-4} \cline{6-7}\cline{9-10}\cline{12-13}
Level $\ell$ && $\tilde{\zeta}(\ell)$ & $\mathbb{P}(L=\ell)$ && $\tilde{\zeta}(\ell)$ & $\mathbb{P}(L=\ell)$ && $\tilde{\zeta}(\ell)$ & $\mathbb{P}(L=\ell)$ &&  $\tilde{\zeta}(\ell)$ & $\mathbb{P}(L=\ell)$  \\
 \hline
9	&&	7.217	&	0.001	&&	5.749	&	0.002	&&	4.917	&	0.004	&&	4.366	&	0.008	\\
8	&&	6.246	&	0.000	&&	4.715	&	0.001	&&	3.889	&	0.002	&&	3.357	&	0.002	\\
7	&&	5.573	&	0.000	&&	4.259	&	0.001	&&	3.540	&	0.002	&&	3.073	&	0.002	\\
6	&&	5.000	&	0.000	&&	3.834	&	0.001	&&	3.202	&	0.002	&&	2.792	&	0.003	\\
5	&&	4.695	&	0.001	&&	3.404	&	0.001	&&	2.854	&	0.003	&&	2.500	&	0.004	\\
4	&&	3.805	&	0.001	&&	2.936	&	0.002	&&	2.479	&	0.004	&&	2.184	&	0.006	\\
3	&&	3.065	&	0.002	&&	2.407	&	0.005	&&	2.060	&	0.008	&&	1.834	&	0.012	\\
2	&&	2.210	&	0.007	&&	1.812	&	0.015	&&	1.591	&	0.023	&&	1.443	&	0.030	\\
1	&&	1.369	&	0.044	&&	1.203	&	0.073	&&	1.101	&	0.095	&&	1.026	&	0.111	\\
0	&&	0.727	&	0.944	&&	0.692	&	0.898	&&	0.665	&	0.858	&&	0.641	&	0.821\\ 
\hline\hline																			
HMSE	&& \multicolumn{2}{c}{14179.63}	&&	\multicolumn{2}{c}{13189.89}&&	\multicolumn{2}{c}{12525.65}&&	\multicolumn{2}{c}{12053.38}	\\ \hline
\end{tabular}
} 

  \bigskip
  \resizebox{0.9\textwidth}{!}{ 
 \begin{tabular}{ l c c c c c c c c c c c c c c c c c c c }
 \multicolumn{12}{l}{(b) $-1/+2/pen$ systems with various $pen$ for $\lambda_0^{[1]}=0.05$}\\
 \hline
&$pen$ &  \multicolumn{2}{c}{0} && \multicolumn{2}{c}{1} && \multicolumn{2}{c}{2} && \multicolumn{2}{c}{3}  \\\cline{3-4} \cline{6-7}\cline{9-10}\cline{12-13}
 Level $\ell$ && $\tilde{\zeta}(\ell)$ & $\mathbb{P}(L=\ell)$ && $\tilde{\zeta}(\ell)$ & $\mathbb{P}(L=\ell)$ && $\tilde{\zeta}(\ell)$ & $\mathbb{P}(L=\ell)$ &&  $\tilde{\zeta}(\ell)$ & $\mathbb{P}(L=\ell)$ \\
 \hline
 9	&&	4.876	&	0.002	&&	4.069	&	0.007	&&	3.582	&	0.013	&&	3.247	&	0.020	\\
8	&&	4.093	&	0.002	&&	3.170	&	0.004	&&	2.647	&	0.005	&&	2.307	&	0.007	\\
7	&&	3.568	&	0.002	&&	2.882	&	0.004	&&	2.491	&	0.005	&&	2.229	&	0.006	\\
6	&&	3.005	&	0.003	&&	2.389	&	0.005	&&	2.045	&	0.008	&&	1.819	&	0.011	\\
5	&&	2.672	&	0.003	&&	2.239	&	0.005	&&	1.985	&	0.007	&&	1.809	&	0.008	\\
4	&&	2.067	&	0.007	&&	1.706	&	0.014	&&	1.509	&	0.020	&&	1.376	&	0.027	\\
3	&&	1.887	&	0.009	&&	1.656	&	0.012	&&	1.510	&	0.013	&&	1.403	&	0.014	\\
2	&&	1.259	&	0.041	&&	1.115	&	0.066	&&	1.029	&	0.086	&&	0.966	&	0.101	\\
1	&&	1.205	&	0.038	&&	1.117	&	0.034	&&	1.053	&	0.030	&&	1.002	&	0.028	\\
0	&&	0.684	&	0.892	&&	0.655	&	0.850	&&	0.631	&	0.812	&&	0.611	&	0.778	\\ 
\hline\hline																
HMSE	&& \multicolumn{2}{c}{13230.45}	&&	\multicolumn{2}{c}{12675.38}&&	\multicolumn{2}{c}{12303.48}&&	\multicolumn{2}{c}{12042.78}	\\ \hline
\end{tabular}
} 

  \bigskip

  \resizebox{0.9\textwidth}{!}{ 
 \begin{tabular}{ l c c c c c c c c c c c c c c c c c c c }
 \multicolumn{12}{l}{(c) $-1/+1/pen$ systems with various $pen$ for $\lambda_0^{[1]}=1$}\\
 \hline
&$pen$ &  \multicolumn{2}{c}{0} && \multicolumn{2}{c}{1} && \multicolumn{2}{c}{2} && \multicolumn{2}{c}{3}  \\\cline{3-4} \cline{6-7}\cline{9-10}\cline{12-13}
 Level $\ell$ && $\tilde{\zeta}(\ell)$ & $\mathbb{P}(L=\ell)$ && $\tilde{\zeta}(\ell)$ & $\mathbb{P}(L=\ell)$ && $\tilde{\zeta}(\ell)$ & $\mathbb{P}(L=\ell)$ &&  $\tilde{\zeta}(\ell)$ & $\mathbb{P}(L=\ell)$ \\
 \hline
14	&&	1.439	&	0.311	&&	1.230	&	0.463	&&	1.123	&	0.562	&&	1.055	&	0.634	\\
13	&&	0.939	&	0.089	&&	0.711	&	0.060	&&	0.590	&	0.045	&&	0.513	&	0.035	\\
12	&&	0.796	&	0.044	&&	0.634	&	0.036	&&	0.539	&	0.030	&&	0.475	&	0.025	\\
11	&&	0.717	&	0.027	&&	0.581	&	0.025	&&	0.500	&	0.022	&&	0.444	&	0.020	\\
10	&&	0.666	&	0.020	&&	0.542	&	0.019	&&	0.468	&	0.018	&&	0.418	&	0.016	\\
9	&&	0.627	&	0.016	&&	0.510	&	0.016	&&	0.442	&	0.015	&&	0.395	&	0.014	\\
8	&&	0.595	&	0.014	&&	0.483	&	0.015	&&	0.419	&	0.014	&&	0.375	&	0.013	\\
7	&&	0.567	&	0.014	&&	0.459	&	0.014	&&	0.397	&	0.014	&&	0.356	&	0.013	\\
6	&&	0.540	&	0.014	&&	0.436	&	0.015	&&	0.377	&	0.014	&&	0.337	&	0.013	\\
5	&&	0.513	&	0.016	&&	0.413	&	0.016	&&	0.357	&	0.016	&&	0.319	&	0.014	\\
4	&&	0.483	&	0.020	&&	0.388	&	0.020	&&	0.335	&	0.018	&&	0.300	&	0.017	\\
3	&&	0.450	&	0.027	&&	0.361	&	0.026	&&	0.313	&	0.024	&&	0.280	&	0.021	\\
2	&&	0.409	&	0.042	&&	0.330	&	0.038	&&	0.287	&	0.033	&&	0.258	&	0.029	\\
1	&&	0.359	&	0.078	&&	0.293	&	0.065	&&	0.256	&	0.053	&&	0.232	&	0.044	\\
0	&&	0.291	&	0.269	&&	0.246	&	0.172	&&	0.219	&	0.122	&&	0.201	&	0.091	\\
\hline\hline	
HMSE	&& \multicolumn{2}{c}{4649638}	&&	\multicolumn{2}{c}{5052378}&&	\multicolumn{2}{c}{5328804}&&	\multicolumn{2}{c}{5532048}	\\ \hline	
\end{tabular}
} 

\end{table}

%
%

%
%

\section{Data analysis}\label{num}


\subsection{Summary of data estimation results}

In this subsection, we summarize the data estimation results using real data. In order to examine the effect of dependence on ratemaking, we utilize a data set concerning collision coverage for new and old vehicles from the Wisconsin Local Government Property Insurance Fund (LGPIF) \citep{Frees4}, where detailed information on the project can be found at the LGPIF project website. 
Such collision coverage provides cover for impact of vehicle with an object, impact of vehicle with an attached vehicle, and overturn of a vehicle. The observations include policyholders who have either new collision coverage or old collision coverage or both. In our data analysis, longitudinal data over the policy years from 2006 to 2010 with 497 governmental entities are used. There are two categorical variables: the entity type with six levels, and the coverage with three levels as shown in Table \ref{tab.x} (which is the same as Table C.6 in \citet{PengAhn}).

\begin{table}[h!t!]
\centering
\caption{(Data analysis) Observable policy characteristics used as covariates} \label{tab.x}\vspace{-.05in}
\begin{tabular}{l|l r r r r r r r }
\hline
Categorical & \multirow{2}{*}{Description} &&  \multicolumn{3}{c}{\multirow{2}{*}{Proportions}} \\
variables &  &  &  &   \\
\hline
Entity type   & Type of local government entity    \\
		& \quad\quad\quad\quad\quad\quad Miscellaneous  	&& \multicolumn{3}{c}{5.03$\%$} \\
		& \quad\quad\quad\quad\quad\quad City			&& \multicolumn{3}{c}{9.66$\%$} \\
		& \quad\quad\quad\quad\quad\quad County			&& \multicolumn{3}{c}{11.47$\%$} \\
		& \quad\quad\quad\quad\quad\quad School			&& \multicolumn{3}{c}{36.42$\%$} \\
		& \quad\quad\quad\quad\quad\quad Town			&& \multicolumn{3}{c}{16.90$\%$} \\
		& \quad\quad\quad\quad\quad\quad Village 			&& \multicolumn{3}{c}{20.52$\%$} \\
\hline
Coverage & Collision coverage amount for old and new vehicles\\
		& \quad\quad\quad\quad\quad\quad Coverage $\in (0,0.14] = 1 $   && \multicolumn{3}{c}{33.40$\%$} \\
		& \quad\quad\quad\quad\quad\quad Coverage $\in (0.14,0.74] = 2 $	&& \multicolumn{3}{c}{33.20$\%$} \\
		& \quad\quad\quad\quad\quad\quad Coverage $\in (0.74,\infty) = 3$	&& \multicolumn{3}{c}{33.40$\%$} \\
\hline
\end{tabular}
\end{table}

Specifically, for Model \ref{mod.1} in the data analysis, as in \eqref{eq.402} we consider
$$
N_{it}|(\Theta_i=\theta_i,\boldsymbol{X}_{i}=\boldsymbol{x}_{i})\sim {\rm Poisson}(\lambda_i\theta_i)\quad\hbox{with}\quad \lambda_i:=\exp(\boldsymbol{x}_i\boldsymbol{\beta})
$$
where $\boldsymbol{x}_i$ is a row vector containing the $i$-th policyholder's information of the categorical variables in Table \ref{tab.x} for the frequency, and the column vector $\boldsymbol{\beta}$ contains the corresponding parameters. For the random effect $\Theta$, we assume a Lognormal distribution with mean 1 and variance of $\exp(\sigma^2)-1$, that is,
$$\Theta \sim {\rm Lognormal}(-\sigma^2/2, \sigma^2).$$
Summary statistics of the posterior samples for the parameters in Model \ref{mod.1} using Bayesian approach are presented in Table \ref{est.model1}. The table includes the posterior median (Est), the posterior standard deviation (Std.dev), and the 95\% highest posterior density Bayesian credible interval (95$\%$ CI). Note that a $*$ sign indicates that the parameters are significant at 0.05 level. In estimating the parameters in Table \ref{est.model1}, we have run 30,000 MCMC iterations saving every 5th sample after burn-in of 15,000 iterations by using JAGS. Standard MCMC diagnostics gave no indication of lack of convergence. For Model \ref{mod.2}, 
in line with \eqref{eq.n} and \eqref{model2Yassump}, we assume
$$N_{it}|(\Theta_i^{[1]}=\theta_i^{[1]},\boldsymbol{X}_{i}^{[1]}=\boldsymbol{x}_{i}^{[1]})\sim {\rm Poisson}(\lambda_i^{[1]}\theta_i^{[1]})$$
for the frequency, and
  \begin{equation*}
       Y_{itj}|(\Theta_i^{[2]}=\theta_i^{[2]}, \boldsymbol{X}_{i}^{[2]}=\boldsymbol{x}_{i}^{[2]}) \sim {\rm Gamma}(\lambda_i^{[2]}\theta_i^{[2]},1/\psi^{[2]})
  \end{equation*}
for the individual claim severity, where $\lambda_i^{[2]}\theta_i^{[2]}$ is the mean and $1/\psi^{[2]}$ is the shape parameter of the Gamma distribution.
With log link, we assume
\[
\lambda_i^{[1]}=\exp(\boldsymbol{x}_i^{[1]}\boldsymbol{\beta}^{[1]})
\quad \hbox{and} \quad
\lambda_i^{[2]}=\exp(\boldsymbol{x}_i^{[2]}\boldsymbol{\beta}^{[2]}),
\]
where we take $\boldsymbol{x}_i^{[1]}=\boldsymbol{x}_i^{[2]}$, and $\boldsymbol{\beta}^{[1]}$ and $\boldsymbol{\beta}^{[2]}$ are the corresponding parameters.

For the bivariate random effect $(\Theta^{[1]},\Theta^{[2]})$, we assume a Gaussian copula $C$ in \eqref{eq.4} with correlation coefficient $\rho$,
and the marginal distributions are assumed to follow Lognormal distributions with different parameters specified as
\[
\begin{cases}
\Theta^{[1]} ~\sim& {\rm Lognormal}(-\sigma_1^2/2, \sigma_1^2),\\
\Theta^{[2]} ~\sim& {\rm Lognormal}(-\sigma_2^2/2, \sigma_2^2).
\end{cases}
\]
The estimation results from \citet{PengAhn} are summarized in Table \ref{est.model2}, and we refer interested readers to \citet{PengAhn} for the details of the estimation procedure and further results. In particular, the dependence parameter of the Gaussian copula is estimated to be $\rho = -0.447$, suggesting a significant negative dependence between the random effects $\Theta^{[1]}$ and $\Theta^{[2]}$ of the frequency and the severity respectively.

\subsection{Analysis of optimal relativities in modified BMS}
\begin{table}[h!]
\caption{(Data analysis) Estimation results under the frequency-only Model 1}\vspace{-.05in}
\centering
\begin{tabular}{ l r r r r r r r  }
 \hline

parameter& Est & Std.dev & \multicolumn{2}{c}{95$\%$ CI} \\
& & & lower & upper \\
 \hline
 \multicolumn{4}{l}{ {\bf Fixed effect}} \\
(Intercept)	&	-2.798	&	0.312	&	-3.433	&	-2.225&*	\\
Type=City		&	0.601	&	0.337	&	-0.044	&	1.280&	\\
Type=County	&	1.923	&	0.329	&	1.257	&	2.543&*	\\
Type=School	&	0.438	&	0.303	&	-0.172	&	1.023&	\\
Type=Town		&	-1.343	&	0.380	&	-2.087	&	-0.600&*	\\
Type=Village	&	-0.005	&	0.320	&	-0.630	&	0.618&	\\
Coverage=2	&	1.254	&	0.214	&	0.820	&	1.651&*	\\
Coverage=3	&	2.156	&	0.231	&	1.688	&	2.587&*	\\
\hline
 \multicolumn{4}{l}{ {\bf Random effect}} \\
$\sigma^2$&  0.993	&	0.14	&	0.733	&	1.276&*	\\
 \hline
\end{tabular}
\label{est.model1}
\end{table}

Using the posterior median as estimates of the parameters in Models \ref{mod.1} and \ref{mod.2} when $z=14$, we calculate the optimal relativity and stationary probability for each BM level as well as the values of HMSE under the modified BM transition rules for a $-1/+1/pen$ system and a $-1/+2/pen$ system when $pen=0,1,2,3$. The results are summarized in Tables \ref{table.3} and \ref{table.4}. We remark that, unlike in Examples \ref{ex.3} and \ref{ex.4} where the stationary probabilities $\mathbb{P}(L=\ell)$ are identical across Tables \ref{table.1} and \ref{table.2}, the stationary probabilities in Tables \ref{table.3} and \ref{table.4} are close but not identical. The reason is that the frequency parameters for Model 2 in Table \ref{est.model2} are estimated jointly with the severity parameters and therefore the results are slightly different from those for Model 1 given in Table \ref{est.model1} where only frequency parameters are estimated. Note also that the parameters are estimated using MCMC sampling which has also possibly contributed to some differences.

Comparing across Tables \ref{table.3} and \ref{table.4}, one can observe the impact of frequency-severity dependence on the optimal relativities. In particular, under negative dependence between frequency and severity, which is implied from the negative dependence between the two random effects with $\rho = -0.447$, the optimal relativities tend to decrease for higher BM levels but increase for lower BM levels. Consequently, the optimal relativities in Table \ref{table.4} are less spread out than those in Table \ref{table.3}. 
It is noted that all optimal BM relativities in Table \ref{table.4} are less than 1, which may be counter-intuitive at a first glance. However, it is noted that, roughly speaking, $\tilde{\zeta}(\ell)$ may be regarded as an estimate of $\Theta^{[1]}\Theta^{[2]}$ when $L^*=(\ell)_a$ for some $a$ (see \eqref{eq.21} and the optimization \eqref{eq.4066} in Model \ref{mod.2}). Since $\Theta^{[1]}$ and $\Theta^{[2]}$ are negatively dependent (with both having mean 1), when one of them is large the other is more likely to be small, making the product unlikely to be larger than 1.

Both Tables \ref{table.3} and \ref{table.4} confirm that as $pen$ increases, each optimal relativity $\tilde{\zeta}(\ell)$ (for fixed $\ell$) decreases and the proportion of policyholders in higher BM levels increases. Explanations similar to those in Examples \ref{ex.3} and \ref{ex.4} are applicable. Certain degree of diversification effects on the BM levels as $pen$ increases resembling Tables \ref{table.1}(a) and (b) in Example \ref{ex.3} and Tables \ref{table.2}(a) and (b) in Example \ref{ex.4} can be observed from Tables \ref{table.3} and \ref{table.4}. Meanwhile, the HMSE values show that a higher $pen$ leads to a decrease in prediction power in this example. The higher HMSE in Table \ref{table.4} compared to Table \ref{table.3} is again attributed to the fact that Model \ref{mod.2} is concerned with the aggregate claim while Model \ref{mod.1} is about claim number, and interested readers are referred to Tables A4 and A5 in \citet{PengAhn} for the magnitude of the average claim severity in Model \ref{mod.2}.

Lastly, we can compare the stationary distribution $L$ of the BM level and the optimal relativities $\tilde{\zeta}(\ell)$ under two different BM systems: (i) a $-1/+1/2$ system with a small increase of BM level per claim ($h=1$) but a period of penalty ($pen=2$); and (ii) a $-1/+2/0$ system with a larger increase of BM level per reported claim ($h=2$) but without period of penalty ($pen=0$). From Tables \ref{table.3}(a) and (b) concerning Model 1, the presence of a period of penalty leads to a smaller proportion of policyholders in the lowest level 0 (with $\mathbb{P}(L=0)=0.505$ in the former model and $\mathbb{P}(L=0)=0.542$ in the latter) but a larger proportion in the highest level 14 (with $\mathbb{P}(L=14)=0.237$ in the former model and $\mathbb{P}(L=14)=0.173$ in the latter). For a given BM level $\ell$, the optimal relativity $\tilde{\zeta}(\ell)$ in the former model is always lower than the latter, suggesting that a BMS with a period of penalty instead of a higher increase of the BM level per claim can look more attractive to the market if the relativities are available to potential customers. Such effect is even more pronounced if one assumes $pen=3$ instead of $pen=2$. Hence, this numerical example clearly illustrates that different outcomes can be anticipated in terms of the distribution of the BM level as well as the optimal relativities depending on how more rigid the BMS is designed. Instead of (or in addition to) imposing a higher increase of the BM level per claim, our model makes it harder for a policyholder to transit to lower BM levels once he/she has reported a claim by requiring consecutive and multiple claim-free years to enjoy a bonus. We also remark that the same conclusion can be drawn from Tables \ref{table.4}(a) and (b) by comparing the $-1/+1/2$ (or $-1/+1/3$) system and the $-1/+2/0$ system under Model 2.




\begin{table}[h!]

  \caption{(Data analysis) Distribution of $L$ and optimal relativities under frequency-only Model \ref{mod.1}}\vspace{-.05in}
  \centering
 \begin{tabular}{ l c c c c c c c c c c c c c c c c c c c }
 \multicolumn{12}{l}{(a) $-1/+1/pen$ system with various $pen$}\\

 \hline
&$pen$ &  \multicolumn{2}{c}{0} && \multicolumn{2}{c}{1} && \multicolumn{2}{c}{2} && \multicolumn{2}{c}{3} \\\cline{3-4} \cline{6-7}\cline{9-10}\cline{12-13}
Level $\ell$ && $\tilde{\zeta}(\ell)$ & $\mathbb{P}(L=\ell)$ && $\tilde{\zeta}(\ell)$ & $\mathbb{P}(L=\ell)$ && $\tilde{\zeta}(\ell)$ & $\mathbb{P}(L=\ell)$ &&  $\tilde{\zeta}(\ell)$ & $\mathbb{P}(L=\ell)$  \\
 \hline
14	&&	1.295	&	0.143	&&	1.167	&	0.197	&&	1.114	&	0.237	&&	1.086	&	0.268	\\
13	&&	0.484	&	0.031	&&	0.362	&	0.022	&&	0.312	&	0.018	&&	0.285	&	0.016	\\
12	&&	0.385	&	0.015	&&	0.321	&	0.014	&&	0.290	&	0.013	&&	0.271	&	0.012	\\
11	&&	0.342	&	0.009	&&	0.298	&	0.010	&&	0.276	&	0.010	&&	0.262	&	0.010	\\
10	&&	0.319	&	0.007	&&	0.284	&	0.008	&&	0.266	&	0.008	&&	0.255	&	0.009	\\
9	&&	0.304	&	0.006	&&	0.274	&	0.007	&&	0.259	&	0.008	&&	0.249	&	0.008	\\
8	&&	0.293	&	0.005	&&	0.266	&	0.007	&&	0.253	&	0.007	&&	0.244	&	0.008	\\
7	&&	0.285	&	0.005	&&	0.260	&	0.007	&&	0.248	&	0.008	&&	0.241	&	0.009	\\
6	&&	0.277	&	0.006	&&	0.255	&	0.007	&&	0.244	&	0.009	&&	0.238	&	0.010	\\
5	&&	0.271	&	0.007	&&	0.251	&	0.009	&&	0.241	&	0.010	&&	0.235	&	0.012	\\
4	&&	0.265	&	0.009	&&	0.247	&	0.012	&&	0.237	&	0.014	&&	0.233	&	0.015	\\
3	&&	0.259	&	0.014	&&	0.242	&	0.018	&&	0.234	&	0.021	&&	0.229	&	0.023	\\
2	&&	0.252	&	0.026	&&	0.237	&	0.034	&&	0.229	&	0.038	&&	0.225	&	0.041	\\
1	&&	0.243	&	0.072	&&	0.230	&	0.087	&&	0.223	&	0.094	&&	0.218	&	0.098	\\
0	&&	0.228	&	0.645	&&	0.217	&	0.562	&&	0.210	&	0.505	&&	0.205	&	0.462	\\
\hline\hline																			
HMSE	&& \multicolumn{2}{c}{2.42280}	&&	\multicolumn{2}{c}{2.54650}&&	\multicolumn{2}{c}{2.60334}&&	\multicolumn{2}{c}{2.63588}	\\ \hline
\end{tabular}

  \bigskip

 \begin{tabular}{ l c c c c c c c c c c c c c c c c c c c }
 \multicolumn{12}{l}{(b) $-1/+2/pen$ system with various $pen$}\\

 \hline
&$pen$ &  \multicolumn{2}{c}{0} && \multicolumn{2}{c}{1} && \multicolumn{2}{c}{2} && \multicolumn{2}{c}{3}  \\\cline{3-4} \cline{6-7}\cline{9-10}\cline{12-13}

 Level $\ell$ && $\tilde{\zeta}(\ell)$ & $\mathbb{P}(L=\ell)$ && $\tilde{\zeta}(\ell)$ & $\mathbb{P}(L=\ell)$ && $\tilde{\zeta}(\ell)$ & $\mathbb{P}(L=\ell)$ &&  $\tilde{\zeta}(\ell)$ & $\mathbb{P}(L=\ell)$ \\
 \hline
14	&&	1.240	&	0.173	&&	1.130	&	0.236	&&	1.088	&	0.280	&&	1.065	&	0.315	\\
13	&&	0.452	&	0.047	&&	0.346	&	0.031	&&	0.301	&	0.024	&&	0.276	&	0.020	\\
12	&&	0.354	&	0.025	&&	0.305	&	0.021	&&	0.278	&	0.018	&&	0.262	&	0.017	\\
11	&&	0.312	&	0.017	&&	0.282	&	0.015	&&	0.265	&	0.015	&&	0.254	&	0.014	\\
10	&&	0.288	&	0.013	&&	0.267	&	0.013	&&	0.255	&	0.013	&&	0.246	&	0.013	\\
9	&&	0.273	&	0.011	&&	0.258	&	0.011	&&	0.248	&	0.012	&&	0.242	&	0.011	\\
8	&&	0.263	&	0.010	&&	0.250	&	0.011	&&	0.242	&	0.013	&&	0.236	&	0.013	\\
7	&&	0.256	&	0.010	&&	0.245	&	0.011	&&	0.239	&	0.011	&&	0.234	&	0.011	\\
6	&&	0.250	&	0.011	&&	0.240	&	0.014	&&	0.233	&	0.016	&&	0.228	&	0.017	\\
5	&&	0.245	&	0.012	&&	0.237	&	0.013	&&	0.232	&	0.013	&&	0.228	&	0.012	\\
4	&&	0.240	&	0.018	&&	0.231	&	0.023	&&	0.225	&	0.027	&&	0.221	&	0.030	\\
3	&&	0.237	&	0.019	&&	0.230	&	0.019	&&	0.225	&	0.017	&&	0.221	&	0.016	\\
2	&&	0.230	&	0.050	&&	0.221	&	0.064	&&	0.216	&	0.072	&&	0.211	&	0.077	\\
1	&&	0.228	&	0.043	&&	0.221	&	0.033	&&	0.217	&	0.026	&&	0.213	&	0.022	\\
0	&&	0.213	&	0.542	&&	0.206	&	0.486	&&	0.201	&	0.444	&&	0.196	&	0.410	\\
\hline\hline	
HMSE	&& \multicolumn{2}{c}{2.48426}	&&	\multicolumn{2}{c}{2.58912}&&	\multicolumn{2}{c}{2.63540}&&	\multicolumn{2}{c}{2.66122}	\\ \hline	
\end{tabular}
\label{table.3}
\end{table}

%
%
%


\begin{table}[h!]
\caption{(Data analysis) Estimation results under the frequency-severity Model \ref{mod.2} with dependence}\vspace{-.05in}
\centering
\begin{tabular}{ l r r r r l r r r r l c c }
 \hline
&& &\multicolumn{2}{c}{95$\%$ CI}& \\
parameter& Est & Std.dev & lower & upper&  \\
 \hline
 \multicolumn{3}{l}{ \textbf{Frequency part}} \\
\quad Intercept &	-2.767& 	0.318&  -3.417&  -2.153&*&\\
\quad City      & 	 0.597& 	0.337&  -0.051& 	1.272& &\\
\quad County    &  1.907& 	0.335& 	1.271& 	2.587&*&	\\
\quad School    &	 0.411& 	0.304&  -0.181& 	1.014& &\\
\quad Town      &	-1.351& 	0.384&  -2.103&  -0.584&*&\\
\quad Village   &	-0.012& 	0.323&  -0.626& 	0.654& &\\
\quad Coverage2 &	 1.247& 	0.212& 	0.829& 	1.667&*& \\
\quad Coverage3 &  2.139& 	0.230& 	1.713& 	2.615&*& 	\\
\hline
 \multicolumn{3}{l}{ \textbf{Severity part}} \\

\quad Intercept &  	 8.829& 	0.375& 	8.103& 	9.588&*& 	 \\
\quad City   	  & 		-0.036& 	0.353&  -0.737& 	0.637& & \\
\quad County 	  & 		 0.341& 	0.338&  -0.336& 	0.980& & \\
\quad School 	  & 		-0.173& 	0.328&  -0.805& 	0.484& & \\
\quad Town  	  &		 0.497& 	0.440&  -0.356& 	1.349& & \\
\quad Village   &		 0.316& 	0.346&  -0.357& 	0.994& & \\
\quad Coverage2 &		 0.180& 	0.244&  -0.308& 	0.646& & 	\\
\quad Coverage3 & 		-0.027& 	0.261&  -0.533& 	0.493& & \\
\quad $1/\psi^{[2]}$&   0.670& 	0.041& 	0.592& 	0.752&*& 	\\
\hline
\multicolumn{3}{l}{ \textbf{Copula part}} \\

\quad $\sigma^2_1$&  	 0.992& 	0.142& 	0.746& 	1.292&*& \\
\quad $\sigma^2_2$&  	 0.293& 	0.067& 	0.176& 	0.433&*&\\
\quad $\rho$&       	-0.447& 	0.130&  -0.690&  -0.190&*\\
 \hline
\end{tabular}
\label{est.model2}
\end{table}

\begin{table}[h!]

  \caption{(Data analysis) Distribution of $L$ and optimal relativities under the frequency-severity Model \ref{mod.2} with dependence}\vspace{-.05in}
  \centering
 \begin{tabular}{ l c c c c c c c c c c c c c c c c c c c }
 \multicolumn{12}{l}{(a) $-1/+1/pen$ system with various $pen$}\\

 \hline
&$pen$ &  \multicolumn{2}{c}{0} && \multicolumn{2}{c}{1} && \multicolumn{2}{c}{2} && \multicolumn{2}{c}{3} \\\cline{3-4} \cline{6-7}\cline{9-10}\cline{12-13}
Level $\ell$ && $\tilde{\zeta}(\ell)$ & $\mathbb{P}(L=\ell)$ && $\tilde{\zeta}(\ell)$ & $\mathbb{P}(L=\ell)$ && $\tilde{\zeta}(\ell)$ & $\mathbb{P}(L=\ell)$ &&  $\tilde{\zeta}(\ell)$ & $\mathbb{P}(L=\ell)$  \\
 \hline
14	&&	0.968	&	0.144	&&	0.888	&	0.199	&&	0.855	&	0.238	&&	0.837	&	0.270	\\
13	&&	0.459	&	0.031	&&	0.356	&	0.022	&&	0.309	&	0.018	&&	0.282	&	0.016	\\
12	&&	0.376	&	0.015	&&	0.318	&	0.014	&&	0.287	&	0.013	&&	0.268	&	0.012	\\
11	&&	0.338	&	0.010	&&	0.295	&	0.010	&&	0.272	&	0.010	&&	0.259	&	0.010	\\
10	&&	0.316	&	0.007	&&	0.280	&	0.008	&&	0.262	&	0.008	&&	0.252	&	0.009	\\
9	&&	0.301	&	0.006	&&	0.270	&	0.007	&&	0.255	&	0.008	&&	0.247	&	0.008	\\
8	&&	0.290	&	0.005	&&	0.262	&	0.007	&&	0.250	&	0.008	&&	0.243	&	0.008	\\
7	&&	0.281	&	0.005	&&	0.256	&	0.007	&&	0.245	&	0.008	&&	0.240	&	0.009	\\
6	&&	0.274	&	0.006	&&	0.251	&	0.007	&&	0.242	&	0.009	&&	0.239	&	0.010	\\
5	&&	0.267	&	0.007	&&	0.247	&	0.009	&&	0.240	&	0.011	&&	0.237	&	0.012	\\
4	&&	0.261	&	0.009	&&	0.244	&	0.012	&&	0.238	&	0.014	&&	0.237	&	0.015	\\
3	&&	0.256	&	0.014	&&	0.241	&	0.018	&&	0.237	&	0.021	&&	0.237	&	0.023	\\
2	&&	0.250	&	0.026	&&	0.239	&	0.034	&&	0.237	&	0.038	&&	0.237	&	0.041	\\
1	&&	0.245	&	0.072	&&	0.238	&	0.088	&&	0.238	&	0.095	&&	0.238	&	0.098	\\
0	&&	0.240	&	0.643	&&	0.237	&	0.560	&&	0.236	&	0.503	&&	0.236	&	0.460	\\
\hline\hline																			
HMSE	&& \multicolumn{2}{c}{91069385}	&&	\multicolumn{2}{c}{95366837}&&	\multicolumn{2}{c}{97457899}&&	\multicolumn{2}{c}{98678459}	\\ \hline
\end{tabular}

  \bigskip

 \begin{tabular}{ l c c c c c c c c c c c c c c c c c c c }
 \multicolumn{12}{l}{(b) $-1/+2/pen$ system with various $pen$}\\

 \hline
&$pen$ &  \multicolumn{2}{c}{0} && \multicolumn{2}{c}{1} && \multicolumn{2}{c}{2} && \multicolumn{2}{c}{3}  \\\cline{3-4} \cline{6-7}\cline{9-10}\cline{12-13}

 Level $\ell$ && $\tilde{\zeta}(\ell)$ & $\mathbb{P}(L=\ell)$ && $\tilde{\zeta}(\ell)$ & $\mathbb{P}(L=\ell)$ && $\tilde{\zeta}(\ell)$ & $\mathbb{P}(L=\ell)$ &&  $\tilde{\zeta}(\ell)$ & $\mathbb{P}(L=\ell)$ \\
 \hline
14	&&	0.933	&	0.174	&&	0.865	&	0.238	&&	0.838	&	0.282	&&	0.824	&	0.316	\\
13	&&	0.431	&	0.047	&&	0.341	&	0.031	&&	0.298	&	0.024	&&	0.275	&	0.021	\\
12	&&	0.348	&	0.025	&&	0.302	&	0.021	&&	0.277	&	0.018	&&	0.262	&	0.017	\\
11	&&	0.309	&	0.017	&&	0.280	&	0.016	&&	0.263	&	0.015	&&	0.254	&	0.014	\\
10	&&	0.286	&	0.013	&&	0.265	&	0.013	&&	0.254	&	0.013	&&	0.247	&	0.013	\\
9	&&	0.271	&	0.011	&&	0.256	&	0.011	&&	0.248	&	0.012	&&	0.244	&	0.011	\\
8	&&	0.260	&	0.010	&&	0.249	&	0.011	&&	0.244	&	0.013	&&	0.241	&	0.014	\\
7	&&	0.253	&	0.010	&&	0.245	&	0.011	&&	0.241	&	0.011	&&	0.240	&	0.011	\\
6	&&	0.248	&	0.011	&&	0.241	&	0.014	&&	0.239	&	0.016	&&	0.239	&	0.017	\\
5	&&	0.244	&	0.012	&&	0.240	&	0.013	&&	0.239	&	0.013	&&	0.239	&	0.013	\\
4	&&	0.241	&	0.018	&&	0.238	&	0.023	&&	0.238	&	0.027	&&	0.239	&	0.030	\\
3	&&	0.240	&	0.019	&&	0.238	&	0.019	&&	0.238	&	0.017	&&	0.239	&	0.016	\\
2	&&	0.238	&	0.050	&&	0.238	&	0.064	&&	0.239	&	0.072	&&	0.239	&	0.077	\\
1	&&	0.238	&	0.043	&&	0.238	&	0.033	&&	0.239	&	0.026	&&	0.239	&	0.022	\\
0	&&	0.237	&	0.540	&&	0.236	&	0.483	&&	0.236	&	0.441	&&	0.235	&	0.407	\\
\hline\hline																			
HMSE	&& \multicolumn{2}{c}{93352237}	&&	\multicolumn{2}{c}{96970265}&&	\multicolumn{2}{c}{98659783}&&	\multicolumn{2}{c}{99617247}	\\ \hline	
\end{tabular}
\label{table.4}
\end{table}

%
%
%
%


\section{Conclusion}\label{conc}

In automobile third-party liability insurance, BMS has been broadly used as a posteriori ratemaking mechanism which helps insurers rate a policyholder's risk more accurately and thus premium can be calculated more fairly to reflect his/her risk. BMS is also designed to stimulate drivers to practise safer driving by adjusting the premium on policy renewal based on the claim history in the previous year. However, a typical BMS immediately offering a reward to policyholders without claim in the previous year may tend to move policyholders easily towards lower BM levels. 
Hence, to resolve the unbalanced issue due to a concentration of policyholders in the lowest BM level, to prevent a quick recovery (in terms of paying lower premium) for those drivers who have a good claim history only for a single period, and to better distinguish between drivers who are consistently good and those who are only temporarily good, in this paper we introduce a ``period of penalty'' to count the number of consecutive claim-free years required to lower the BM level. Although other more rigid BMS transition rules such as a $-1/+h$ system with $h=2$ or $h=3$ to severely penalize claims in the previous year may put less pressure on the premium income of the insurer, 
these may weaken the product's competitiveness in the market and the insurers may lose better customers. Through numerical illustrations with a data set, we have demonstrated that, compared to a $-1/+2$ system, the introduction of a penalty period to a $-1/+1$ system can result in (i) lower values of optimal BM relativities which can potentially improve marketability of the product; and (ii) better separation of risks because the BM level 0 is less concentrated as drivers who are not consistently good are moved to higher BM levels. With the policyholders aware of the requirement of consecutive claim-free years to enjoy lower BM relativities at lower BM levels, they can be well motivated to drive more safely. 

Moreover, in the afore-mentioned modified BMS with a penalty period, we take into account the randomness of (i) frequency only or (ii) both frequency and severity when modeling the unobserved risk characteristics and to construct the optimal relativities associated with the BM levels. These relativities are then used to determine the premium actually charged to the policyholders staying in the corresponding BM levels. We hope that our extended BMS allowing for dependency between the frequency and the severity of claims via dependency of the unobserved risk characteristics provides some insights for improving the classical BMS, in particular, when the insurers have the freedom to alter the features of the BMS to reflect claim experience more accurately based on the claim history of policyholders in multiple years with product competitiveness in mind.


\vspace{+.2in}

\bibliographystyle{apalike}


\vspace{+.2in}

\appendix

\section{Generalized linear models (GLMs)}\label{EDF}


The {\it exponential dispersion family} (EDF) in \citet{Nelder1989} can be considered for modeling the random components of the frequency and the severity of insurance claims in the GLMs. The EDF with mean $\mu$ and dispersion $\psi$, whose distribution function is denoted by $F(\cdot;\mu, \psi)$, has the probability density/mass function (in $y$)
\[
p(y|\vartheta,\psi)=\exp[(y\vartheta-b(\vartheta))/\psi+c(y,\psi)],
\]
where $\vartheta$ is the canonical parameter, and $b(\cdot)$ and $c(\cdot)$ are predetermined functions. The mean of the distribution can be expressed as $\mu=b^\prime(\vartheta)$ 
and the variance is $b''(\vartheta)\psi\equiv V(\mu)$, 
where the inverse of $b'(\cdot)$ is known as the canonical link function and $V(\cdot)$ is called the variance function.\vspace{-.1in}

\section{Model \ref{mod.1}}\label{model1}
In the frequency random effect model:\vspace{-.05in}
\begin{enumerate}
    \item[i.] For the $i$-th policyholder, the conditional distribution of the number of claims $N_{it}$ in the $t$-th policy year given the observed risk characteristics $\boldsymbol{X}_{i}=\boldsymbol{x}_{i}$ and unobserved risk characteristics  $\Theta_i=\theta_i$ is specified as
       \begin{equation}\label{eq.402}
        N_{it}|(\Theta_i=\theta_i,\boldsymbol{X}_{i}=\boldsymbol{x}_{i}) \iid {F}(\cdot ; \lambda_{i}\theta_i, \psi ), 
        \end{equation}
   where the distribution function $F$ has mean parameter $\lambda_{i}\theta_i$ with $\lambda_{i}= \eta^{-1}(\boldsymbol{x}_{i}\boldsymbol{\beta})$ and some parameter $\psi$. The ``i.i.d.'' in \eqref{eq.402} means that, conditional on $\Theta_i=\theta_i$ and $\boldsymbol{X}_{i}=\boldsymbol{x}_{i}$, the claim numbers $N_{it}$'s for different $t$'s are independent and identically distributed. When $F$ is in the class of EDF in \ref{EDF}, $\psi$ is the dispersion parameter.

     \item[ii.] The random variable $\Theta_i$ for the $i$-th policyholder's unobserved risk characteristics concerning the claim frequency is assumed to be independent of the observed risk characteristics $\boldsymbol{X}_{i}$. The $\Theta_i$'s are i.i.d. having distribution function $G$ and density $g=G'$, and we write
   \begin{equation}\label{eq.403}
  \Theta_i\iid G.
  \end{equation}
A generic variable with distribution $G$ is denoted by $\Theta$, and we assume $\mathbb{E}[\Theta] = 1 $ for convenience.

\end{enumerate}\vspace{-.2in}

\section{Model \ref{mod.2}}\label{model2}

To account for various forms of dependence such as frequency-frequency, severity-severity and frequency-severity, we consider a copula-based random effect model as in \cite{PengAhn}:\vspace{-.05in}
\begin{enumerate}
    \item[i.] The frequency component $N_{it}$ is specified using a count regression model conditional on the observed risk characteristics $\boldsymbol{X}_{i}^{[1]}$  
    and the unobserved risk $\Theta_i^{[1]}$ such that
       \begin{equation}\label{eq.n}
        N_{it}|(\Theta_i^{[1]}=\theta_i^{[1]},\boldsymbol{X}_{i}^{[1]}=\boldsymbol{x}_{i}^{[1]}) \iid {F}_1(\cdot ; \lambda_{i}^{[1]}\theta_i^{[1]}, \psi^{[1]}), 
        \end{equation}
   where the distribution function $F_1$ has mean parameter $\lambda_{i}^{[1]}\theta_i^{[1]}$  with $\lambda_{i}^{[1]}= \eta_1^{-1}(\boldsymbol{x}_{i}^{[1]}\boldsymbol{\beta}^{[1]})$ and some parameter $\psi^{[1]}$. The ``i.i.d.'' in \eqref{eq.n} means that $N_{it}$'s for different $t$'s are (conditionally) independent and identically distributed. If $F_1$ belongs to EDF, then $\psi^{[1]}$ corresponds to the dispersion parameter. Here we implicitly assume that the conditional distribution in \eqref{eq.n} does not change even if we are given further information on the risk characteristics $ \boldsymbol{X}_{i}^{[2]} $ and $\Theta_i^{[2]}$ pertaining to the severity.

   \item[ii.] The distribution of the $i$-th policyholder's $j$-th claim severity $Y_{itj}$ in the $t$-th year conditional on the observed risk characteristics $ \boldsymbol{X}_{i}^{[2]} $ and the unobserved risk characteristics $\Theta_i^{[2]}$ is specified as
  \begin{equation}\label{model2Yassump}
       Y_{itj}|(\Theta_i^{[2]}=\theta_i^{[2]}, \boldsymbol{X}_{i}^{[2]}=\boldsymbol{x}_{i}^{[2]}) \iid {F}_2(\cdot\,;\, \lambda_{i}^{[2]}\theta_i^{[2]}, \psi^{[2]}),
  \end{equation}
where the distribution function $F_2$ is in EDF with mean parameter $\lambda_{i}^{[2]}\theta_i^{[2]}$ 
     (where $\lambda_{i}^{[2]}=\eta_2^{-1}(\boldsymbol{x}_{i}^{[2]} \boldsymbol{\beta}^{[2]})$) and dispersion parameter $\psi^{[2]}$. The ``i.i.d.'' in \eqref{model2Yassump} means that $Y_{itj}$'s for different $t$'s and different $j$'s are (conditionally) independent and identically distributed. It is understood that the conditional distribution in \eqref{model2Yassump} remains the same even if we are further given the risk characteristics $ \boldsymbol{X}_{i}^{[1]} $ and $\Theta_i^{[1]}$ pertaining to the frequency. By the property of EDF, the distribution of the average claim amount $M_{it}$ (see \eqref{sumaverageDef}) conditional on the observed risk characteristics $ \boldsymbol{X}_{i}^{[2]} $, the unobserved risk characteristics $\Theta_i^{[2]}$, and the frequency $N_{it}$ is specified via
  \begin{equation}\label{eq.s}
       M_{it}|(N_{it}=n_{it}, \Theta_i^{[2]}=\theta_i^{[2]}, \boldsymbol{X}_{i}^{[2]}=\boldsymbol{x}_{i}^{[2]}) \iid {F}_2(\cdot\,;\, \lambda_{i}^{[2]}\theta_i^{[2]}, \psi^{[2]}/n_{it}), \qquad n_{it}>0,
  \end{equation}
with updated dispersion parameter $\psi^{[2]}/n_{it}$, and
  \[
       \mathbb{P}(M_{it}=0|N_{it}=n_{it}, \Theta_i^{[2]}=\theta_i^{[2]}, \boldsymbol{X}_{i}^{[2]}=\boldsymbol{x}_{i}^{[2]})=1, \qquad  n_{it}=0.
  \]

The ``i.i.d.'' in \eqref{eq.s} means that $M_{it}$'s for different $t$'s are (conditionally) independent and identically distributed.


   \item[iii.] The bivariate random vector of unobserved risk characteristics $(\Theta_i^{[1]},\Theta_i^{[2]})$ of the $i$-th policyholder, which is assumed to be independent of the observed risk characteristics $(\boldsymbol{X}_{i}^{[1]}, \boldsymbol{X}_{i}^{[2]})$, has joint distribution function $H$ specified via
   \begin{equation}\label{eq.4}
  (\Theta_i^{[1]}, \Theta_i^{[2]})\iid H=C(G_1, G_2), 
  \end{equation}
  where $G_1$ and $G_2$ denote the marginal distribution functions of $\Theta_i^{[1]}$ and $\Theta_i^{[2]}$ respectively, and $C$ is a bivariate copula. The ``i.i.d.'' in \eqref{eq.4} means that $(\Theta_i^{[1]}, \Theta_i^{[2]})$'s for different $i$'s are independent and identically distributed. A generic pair of $(\Theta_i^{[1]},\Theta_i^{[2]})$ will be denoted by $(\Theta^{[1]},\Theta^{[2]})$. We shall use $g_1$, $g_2$, and $h$ to denote the density versions of $G_1$, $G_2$, and $H$, respectively.
For convenience, it is further assumed that
\[
       \mathbb{E}[\Theta^{[1]}] = 1
        \quad\hbox{and}\quad \mathbb{E}[\Theta^{[2]}] = 1.
  \]

\end{enumerate}

\end{document}